%% file: main.tex
\documentclass[a4paper,UKenglish,cleveref, thm-restate]{lipics-v2021}

\hideLIPIcs  

\usepackage[dvipsnames,svgnames,table]{xcolor}
\usepackage{url}

\usepackage{mathtools}
\usepackage{bm}
\usepackage{tikz}
\usetikzlibrary{tikzmark}

\newtheorem{question}{Question}
\newenvironment{bfDescript}[1]{{\sffamily\normalsize\bfseries#1}}

\usepackage[notion, quotation, electronic]{knowledge}

\input{knowledge.kl}

\usepackage{todonotes}

\input{macros}

\title{The memory of $\omega$-regular and $\BCSigma$ objectives} 

\author{Antonio Casares}{University of Kaiserslautern-Landau, Germany \and \url{https://antonio-casares.github.io/}}{antonio.casares@rptu.de}{https://orcid.org/0000-0002-6539-2020}{Partially supported by Deutsche Forschungsgemeinschaft (grant number 522843867) and European Research Council (grant number 101089343).
Part of this work was done while Casares was at the University of Warsaw, Poland, supported by the Polish National Science Centre (NCN) grant ``Polynomial finite state computation'' (2022/46/A/ST6/00072).}

\author{Pierre Ohlmann}{CNRS, Laboratoire d'Informatique et des Systèmes (LIS), Marseille, France\and \url{https://pageperso.lis-lab.fr/pierre.ohlmann/}}{pierre.ohlmann@lis-lab.fr}{https://orcid.org/0000-0002-4685-5253}{}

\authorrunning{A. Casares and P. Ohlmann} 

\Copyright{Antonio Casares and Pierre Ohlmann} 

\ccsdesc[500]{Theory of computation~Logic and verification}

\keywords{Infinite duration games, memory, omega-regular} 

\category{} 

\relatedversion{} 

\acknowledgements{We thank Nathan Lhote for his participation in the scientific discussions which
led to this paper. We also thank Pierre Vandenhove for pointing us to references concerning lifting
results for memory.} 

\nolinenumbers 

\begin{document}

\maketitle

\begin{abstract}
    In the context of 2-player zero-sum infinite duration games played on (potentially infinite) graphs, the memory of  an objective is the smallest integer $k$ such that in any game won by Eve, she has a strategy with $\leq k$ states of memory.
    For $\omega$-regular objectives, checking whether the memory equals a given number $k$ was not known to be decidable.
    In this work, we focus on objectives in $\BCSigma$, i.e. recognised by a potentially infinite deterministic parity automaton. We provide a class of automata that recognise objectives with memory $\leq k$, leading to the following results:
    \begin{itemize}
    \item for $\omega$-regular objectives, the memory over finite and infinite games coincides and can be computed in $\NP$;
    \item given two objectives $W_1$ and $W_2$ in $\BCSigma$ and assuming $W_1$ is prefix-independent, the memory of $W_1 \cup W_2$ is at most the product of the memories of $W_1$ and $W_2$.
    \end{itemize}
Our results also apply to chromatic memory, the variant where strategies can update their "memory state" only depending on which colour is seen.
\end{abstract}

\vspace{-6mm}

\paragraph*{}
This document contains hyperlinks.
\AP Each occurrence of a "notion" is linked to its ""definition"".
On an electronic device, the reader can click on words or symbols (or just hover over them on some PDF readers) to see their definition.

\section{Introduction}
\input{intro.tex}

\section{Preliminaries}
\input{preliminaries}

\section{Characterisation of objectives in $\BCSigma$ with memory $\leq k$}
\input{mainCharacterisation}

\section{Union of objectives}\label{sec:union}
\input{union.tex}

\section{Conclusions and open questions}
\input{conclusion.tex}


\bibliography{bib}

\newpage
\appendix

\section{Proof of the structuration theorem}\label{app:structuration}
\input{appendix-structuration.tex}



\end{document}

%% file: macros.tex
\newcommand{\re}[1]{\xrightarrow{#1}}
\newcommand{\rer}[1]{\xleftrightarrow{#1}}

\newcommand{\rp}[1]{\overset{#1}{\rightsquigarrow}}

\newcommand{\tand}{\text{ and }}
\newcommand{\tor}{\text{ or }}
\newcommand{\tin}{\text{ in }}
\newcommand{\tif}{\text{ if }}

\newcommand{\tfor}{\text{ for }}

\newrobustcmd{\TL}{\kl[\TL]{\text{TL}}}
\newrobustcmd{\TW}{\kl[\TW]{\text{TW}}}
\newrobustcmd{\Parity}{\kl[\Parity]{\mathrm{Parity}}}
\newrobustcmd{\MaxParity}{\kl[\MaxParity]{\mathrm{MaxParity}}}
\knowledge{\TL}{notion}
\knowledge{\TW}{notion}
\knowledge{\Parity}{notion}
\knowledge{\MaxParity}{notion}

\renewcommand{\epsilon}{\varepsilon}

\newcommand{\eps}{\varepsilon}

\newcommand{\odd}{\mathrm{odd}}

\newcommand{\N}{\mathbb N}
\newcommand{\Z}{\mathbb Z}

\newcommand{\fin}[1]{W_{\mathrm{fin}}}

\newrobustcmd{\coBuchi}{\kl[\coBuchi]{\text{co-Büchi}}}
\knowledge{\coBuchi}{notion}

\newrobustcmd{\Weps}{\kl[\Weps]{W^\eps}}
\knowledge{\Weps}{notion}

\newcommand{\boldclass}[3]{\texorpdfstring{\ensuremath{\mathbf{#1}^{#2}_{#3}}}{Borel}}

\newcommand{\bsigma}[1]{\boldclass{\Sigma}{0}{#1}}
\newcommand{\bpi}[1]{\boldclass{\Pi}{0}{#1}}
\newcommand{\bdelta}[1]{\boldclass{\Delta}{0}{#1}}

\newrobustcmd{\BCSigma}{\kl[\BCSigma]{\mathrm{BC}(\boldclass{\Sigma}{0}{2})}}
\knowledge{\BCSigma}{notion}

\newrobustcmd{\Skappad}{\kl[\Skappad]{S^\kappa_{d}}}
\newrobustcmd{\Sgraph}{\kl[\Sgraph]S}
\knowledge{\Skappad}[\Sgraph]{notion}

\newcommand{\A}{\mathcal{A}}

\newrobustcmd{\casc}{\mathrel{\kl[\casc]\ltimes}}

\knowledge{\casc}{notion}

\newcommand{\NP}{\textbf{NP}}

\newcommand{\Sigmaeps}{\Sigma\ {\cup} \left\{\eps\right\}}

\newcommand{\ph}{\_}

\newcommand{\B}{\mathcal B}

\newcommand{\T}{\mathcal T}

\newcommand{\Eve}{\text{Eve}}
\newcommand{\Adam}{\text{Adam}}

\renewcommand{\d}{d}

\newcommand{\done}{[d_1]}
\newcommand{\dtwo}{[d_2]^*}

\newcommand{\doneodd}{[d_1]_{\mathrm{odd}}}
\newcommand{\dtwoodd}{[d_2]^*_{\mathrm{odd}}}

\newrobustcmd{\xodd}{\kl[\xodd]x_{\mathrm{odd}}}
\newrobustcmd{\dodd}{\kl[\dodd]d_{\mathrm{odd}}}

\knowledge{\xodd}[\dodd]{notion}

\newcommand{\oo}{\omega}
\renewcommand{\SS}{\Sigma}

\newcommand{\cleq}{\preccurlyeq}

\newcommand{\infOften}{{\mathtt{Inf}}}
\newcommand{\finOften}{\mathtt{Fin}}
\newcommand{\noOcc}{\mathtt{No}}

\newrobustcmd{\indtau}[1]{\kl[\indtau]{\mathsf{ind}}_\tau(#1)}
\newrobustcmd{\indextau}[2]{\kl[\indtau]{\mathsf{ind}}_{#1}(#2)}
\knowledge{\indtau}[\indextau]{notion}

\newcommand{\ymin}{y_{\mathrm{min}}}

\newcommand{\tmin}{t_{\mathrm{min}}}
\newcommand{\imin}{i_{\mathrm{min}}}

\newcommand{\xbreak}{x^{(0)}}

\newcommand{\powne}[1]{\mathcal P_{\neq \varnothing}(#1)}

\renewcommand{\ul}[1]{\bm{#1}}

\newcommand{\us}[2]{\underset{(#2)}{#1}}

%% file: intro.tex
\subsection*{Context: Strategy complexity in infinite duration games}

We study infinite duration games on graphs in which two players, called Eve and Adam, interact by  moving a token along the edges of a (potentially infinite) edge-coloured directed graph. Each vertex belongs to one player, who chooses where to move next during a play. This interaction goes on for an infinite duration, producing an infinite path in the graph. The winner is determined according to a language of infinite sequences of colours~$W$, called the objective of the game; Eve aims to produce a path coloured by a sequence in $W$, while Adam tries to prevent this.
This model is widespread for its use in verification and synthesis~\cite{HandbookModelChecking2018}.

In order to achieve their goal, players use strategies, which are representations of the course of all possible plays together with instructions on how to act in each scenario.
In this work, we are interested in optimal strategies for Eve, that is, strategies that guarantee a victory whenever this is possible. More precisely, we are interested in the complexity of such strategies, or in other words, in the succinctness of the representation of the space of plays.

\subparagraph*{Positionality.}
The simplest strategies are those that assign in advance an outgoing edge to each vertex owned by Eve, and always play along this edge, disregarding all the other features of the play.
All the information required to implement such a strategy appears in the game graph itself.
\AP Objectives for which such strategies are sufficient to play optimally are called ""positional"" (or memoryless).
Understanding positionality has been the object of a long line of research. The landmark results of Gimbert and Zielonka~\cite{GZ05} and Colcombet and Niwi\'nski~\cite{CN06} gave a good understanding of which objectives are bi-positional, 
i.e.~positional for both players.

More recently, Ohlmann proposed to use universal graphs as a tool for studying positionality (taking  Eve's point of view)~\cite{Ohlmann23}.
This led to many advances in the study of positionality~\cite{BCRV24HalfJournal,OS24Sigma2}, and most notably, a characterisation of positional "$\omega$-regular" objectives by Casares and Ohlmann~\cite{CO24Positional}, together with a polynomial time decision procedure (and some other important corollaries, more discussion below).

\subparagraph*{Strategies with memory.} 
In many scenarios, playing optimally requires distinguishing plays that end in the same vertex.
A seminal result of B\"uchi and Landweber~\cite{BL69Strategies} states that in finite games where the objective is an "$\omega$-regular" language, the winner has a winning strategy that can be implemented by a finite automaton processing the edges of the game; this result was later extended to infinite game graphs by Gurevich and Harrington~\cite{Gurevich1982trees}.
Here, the states of the automaton are interpreted as memory states of the strategy, and a natural measure of the complexity of a strategy is the number of such states.
More precisely, the memory of an objective $W$ is the minimal $k$ such that whenever Eve wins a game with objective $W$, she has a winning strategy with $k$ states of memory.
For "$\omega$-regular" objectives, this is always finite~\cite{BL69Strategies,Gurevich1982trees}, while the case of positionality discussed above corresponds to memory $k=1$.

Characterising the memory of objectives has been a recurrent research subject since the 1990s.
Notable results include the characterisation for Muller objectives by Dziembowski, Jurdzi\'nski, and Walukiewicz~\cite{DJW1997memory}, or for closed objectives by Colcombet, Fijalkow and Horn~\cite{CFH14}.
However, these are all rather restricted classes of "$\omega$-regular" objectives.
The problem of deciding the memory of "$\omega$-regular" objectives has been raised in many occasions (see for instance~\cite[Sect.~9.2]{Kop08Thesis}, \cite[Sect.~4]{BRV22Colours}, \cite[Conclusions]{CFH22TenYears}, or~\cite{BFRV23Regular}),
but prior to this work, even computing the memory of open "$\omega$-regular" objectives was not known to be decidable.

\subparagraph*{Chromatic memory.}
In the special case where the automata implementing strategies are only allowed to read the colours on the edges of the game graph, instead of the edges themselves, we speak of chromatic memory.
In his PhD thesis, Kopczy\'nski showed that, for "prefix-independent" "$\omega$-regular" objectives and over finite game graphs, the "chromatic memory" can be computed in exponential time~\cite[Theorem~8.14]{Kop08Thesis}.
Recently, it was shown that computing the "chromatic memory" of some restricted subclasses of "$\omega$-regular" objectives is in fact $\NP$-complete: for Muller objectives~\cite{Casares22Chromatic} and for topologically open or closed objectives~\cite{BFRV23Regular}.
However, the chromatic and the unconstrained memory of objectives may differ, even exponentially~\cite{Casares22Chromatic,CCL22SizeGFG}.

\subparagraph{Finite-to-infinite lift.} In general, the memory of an objective may differ if we consider only finite game graphs or arbitrary ones (a well-known such example is the mean-payoff objective\footnote{Ohlmann and Skrzypczak recently showed that if defined as $\{w \in \Z^\omega \mid \limsup_k \sum_{i=0}^{k-1} w_i < 0\}$, the mean-payoff objective is even positional over infinite games~\cite{OS24Sigma2}. However, the complement of this objective is only positional on finite game graphs.}, or the unboundedness objective $\{w_0 w_1 \dots \in \{-1,1\}^\omega \mid \forall N, \exists k, \sum_{i=0}^k w_i \geq N\}$).
In his PhD thesis, Vandenhove conjectured that this is not the case for "$\omega$-regular" objectives~\cite[Conjecture~9.1.2]{Vandenhove23Thesis}\footnote{Formally, the conjecture is stated for arena-independent memories (a slightly different setting).}: their memory is the same over finite and infinite games. 
This conjecture has recently been proved in the case of positional (memoryless) objectives, i.e., those with memory equal to one~\cite[Theorem~3.4]{CO24Positional}.

\subparagraph{Unions of objectives.} The driving question in Kopczyński's PhD thesis~\cite{Kop08Thesis} is whether "prefix-independent" "positional" "objectives" are closed under union, which has become known as Kopczyński's conjecture. Recently, Kozachinskiy~\cite{Kozachinskiy24EnergyGroups} disproved this conjecture, but only for positionality over finite game graphs, and using non-$\oo$-regular objectives. In fact, the conjecture is now known to hold for "$\omega$-regular" objectives~\cite{CO24Positional} and $\bsigma 2$ objectives~\cite{OS24Sigma2}. 
Casares and Ohlmann proposed a generalisation of this conjecture from positional objectives to objectives requiring memory~\cite[Conjecture~7.1]{CO25LMCS} (see also~\cite[Proposition~8.11]{Kop08Thesis}):

\begin{conjecture}[{""Generalised Kopczyński's conjecture""}]\label{conj:Kopcz-Union-Memory}
	Let $W_1,W_2\subseteq \SS^\oo$ be two "prefix-independent" "objectives" with "memory" $k_1$ and $k_2$, respectively. Then $W_1\cup W_2$ has "memory" at most $k_1 \cdot k_2$.
\end{conjecture}

Using the characterisation of~\cite{DJW1997memory}, it is not hard to verify that the conjecture holds for Muller objectives.

\subparagraph*{$\BCSigma$ languages.}

The results in this work apply not only to "$\omega$-regular" languages, but to the broader class of $\BCSigma$ languages.
These are boolean combinations of languages in $\bsigma 2$ (countable unions of closed languages), or equivalently, recognised by deterministic parity automata with possibly infinitely many states.
This class includes typical non-$\omega$-regular examples such as energy or mean-payoff objectives, but also broader classes such as unambiguous $\oo$-petri nets~\cite{FSJLS22} and deterministic $\oo$-Turing machines (Turing machines with a Muller condition).

\subsection*{Contributions}

Our main contribution is a characterisation of $\BCSigma$ "objectives" with "memory" $\leq k$, stated in Theorem~\ref{thm:main-charac}.
It captures both the memory and the "chromatic memory" of "objectives" over infinite game graphs.
The characterisation is based on the notion of $k$-wise $\epsilon$-completable automata, which are parity automata with states partitioned in $k$ chains, where each chain is endowed with a tight hierarchical structure encoded in the $\eps$-transitions of the automaton.

From this characterisation, we derive the following corollaries:
\begin{enumerate}
	\item \bfDescript{Decidability in $\NP$.} Given a "deterministic" "parity automaton" $\A$, the "memory" (resp. "chromatic memory") of $L(\A)$  can be computed in $\NP$.
	
	\item \bfDescript{Finite-to-infinite lift.} If an "$\omega$-regular" objective has "memory" (resp. "chromatic memory") $\leq k$ over finite games, then the same is true over arbitrary games. 

	\item \bfDescript{"Generalised Kopczyński's conjecture".} We establish (and strengthen) Conjecture~\ref{conj:Kopcz-Union-Memory} in the case of $\BCSigma$ objectives: if $W_1$ and $W_2$ are $\BCSigma$ objectives with "memory" $k_1$ and
	 $k_2$, and one of them is "prefix-increasing", then the "memory" of $W_1 \cup W_2$ is $\leq k_1\cdot k_2$.
\end{enumerate}

\subparagraph{Our toolbox: Universal graphs and $\eps$-completable automata.}
As mentioned above, Ohlmann proposed a characterisation of "positionality" by means of "universal graphs"~\cite{Ohlmann23}.
In 2023, Casares and Ohlmann extended this characterisation to "objectives" with "memory" $\leq k$ by considering partially ordered "universal graphs"~\cite{CO25LMCS}.
Until now, "universal graphs" have been mainly used to show that certain "objectives" have "memory" $\leq k$ (usually for $k=1$); this is done by constructing a "universal graph" for the "objective".
One technical novelty of this work is to exploit both directions of the characterisation, as we also rely on the existence of "universal graphs" to obtain decidability results.

Our characterisation is based on the notion of "$k$-wise $\eps$-completable" automata, which extends the key notion of~\cite{CO24Positional} from positionality to finite "memory".  

\subparagraph*{Comparison with~\cite{CO24Positional}.}
In 2024, Casares and Ohlmann characterised positional "$\omega$-regular" objectives~\cite{CO24Positional}, establishing decidability of positionality in polynomial time, and settling Kopczyński's conjecture for "$\omega$-regular" objectives.
Although the current paper generalises most of these results to the case of memory, as well as potentially infinite automata, the proof techniques are significantly different: while~\cite{CO24Positional} is based on intricate successive transformations of parity automata, the proof we present here is based on an extraction method in the infinite and manipulates ordinal numbers.
Though somewhat less elementary, the new proof is notably shorter, and probably easier to read.

Still, when instantiated to the case of "memory" $1$, our results recover and extend several of those from~\cite{CO24Positional}:
\begin{itemize}
	\item The finite-to-infinite lift of positionality for "$\omega$-regular" objectives is recovered.
	\item Kopczyński's conjecture is extended to $\BCSigma$ objectives.
	\item Although our computability results only talk about $\NP$, decidability of positionality in polynomial time can be recovered using a relatively simple greedy argument presented in~\cite[Theorem~5.3]{CO24Positional}, which relies on the closure under union, which we do re-establish (previous item).
\end{itemize}
Results that we do not recover include the 1-to-2-player lift and the closure of positionality under addition of neutral letters.
However, the fact that we do not obtain the 1-to-2-player lift is not surprising since it does not hold for memory $k > 1$ (see Proposition~\ref{prop:1-2-player-counterexample}). 

%% file: preliminaries.tex
We let $\Sigma$ be a countable alphabet\footnote{We restrict our study to countable alphabets, as if $\SS$ is uncountable, the topological space $\SS^\omega$ is not Polish and the class $\BCSigma$ is not as well-behaved.} and $\eps\notin \Sigma$ be a fresh symbol that should be interpreted as a neutral letter.
Given a word $w \in (\Sigma \cup \{\eps\})^\omega$ we write $\pi_\Sigma(w)$ for the (finite or infinite) word obtained by removing all $\eps$'s from $w$;  we call $\pi_\Sigma(w)$ the projection of $w$ on $\Sigma$.
\AP An ""objective"" is a set $W \subseteq \Sigma^\omega$.
Given an objective $W \subseteq \Sigma^\omega$, we let $\intro*\Weps$ denote $\pi_\Sigma^{-1}(W) \subseteq (\Sigma \cup \{\eps\})^\omega$.

Throughout the paper, we use Von Neumann's notation for ordinals: $\lambda$ denotes the set of ordinals $<\lambda$.
This also applies to finite numbers, e.g. $k=\{0,\dots,k-1\}$.
As is standard, we also identify cardinals with their initial ordinals.

\subsection{Graphs, games and memory}

We introduce notions pertaining to games and strategy complexity, as they will be central in the statement of our results.
Nevertheless, we note that all our technical proofs will use these definitions through Theorem~\ref{thm:universal_graphs} below, and will not explicitly use games.

\subparagraph*{Graphs.}
\AP A ""$\SS$-graph"" $G$ is given by a set of vertices $V(G)$ and a set of coloured, directed edges $E(G) \subseteq V(G) \times \SS \times V(G)$.
We write $v \re c v'$ for edges $(v,c,v')$.
A path is a sequence of edges with matching endpoints $(v_0 \re {c_0} v_1)(v_1 \re{c_0} v_2) \dots$ which we write as $v_0 \re {c_0} v_1 \re{c_1} \dots$.
Paths can be empty, finite, or infinite, and have a label $c_0c_1\dots$.
Throughout the paper, graphs are implicitly assumed to be without dead-end: every vertex has an outgoing edge.

\AP We say that a vertex $v$ in a $\Sigma$-graph (resp. a $(\Sigma \cup \{\eps\})$-graph) ""satisfies"" an objective $W \subseteq \Sigma^\omega$ if the label of any infinite path from $v$ belongs to $W$ (resp. to $\Weps$).
\AP A ""pointed graph"" is a graph with an identified initial vertex.
A "pointed graph" satisfies an objective $W \subseteq \Sigma^\omega$ if the initial vertex satisfies $W$; a non-pointed graph satisfies an objective if all its vertices do.
\AP An ""infinite tree"" is a sinkless "pointed graph" whose initial vertex is called the root, and with the property that every vertex admits a unique path from the root.

\AP A ""morphism"" from a $\Sigma$-graph $G$ to a $\Sigma$-graph $H$ is a map $\phi: V(G) \to V(H)$ such that for any edge $v \re c v'$ in $G$, it holds that $\phi(v) \re c \phi(v')$ is an edge in $H$.
Morphisms between pointed graphs should moreover send the initial vertex of $G$ to the initial vertex of $H$.
Morphisms need not be injective.
We write 
$G \re{} H$ when there exists a morphism $\phi \colon G \to H$.

\subparagraph*{Games and strategies.}
\AP A ""game"" is given by a "pointed" $(\Sigma \cup \{\eps\})$-graph $G$ together with an objective $W \subseteq \Sigma^\omega$, and a partition of the vertex set $V(G)=V_\Eve \sqcup V_\Adam$ into the vertices controlled by Eve and those controlled by Adam.
\AP A ""strategy""\footnote{We follow the terminology from~\cite{CO25LMCS}. The classical notion of a strategy as a function $f\colon E(G)^*\to V(G)$ can be recovered by considering the graph with vertices $E(G)^*$, and edges $\rho \re e \rho e$.} (for Eve) is a "pointed graph" together with a morphism $\pi$ towards $G$, 
satisfying that for every edge $v \re c v'$ in the game, where $v \in V_\Adam$, and for all $u \in \pi^{-1}(v)$, there is an edge $u \re c u'$ such that $u' \in \pi^{-1}(v')$.
\AP A strategy is ""winning"" if it "satisfies" the "objective" $W$ of the game. 
We say that Eve wins if there exists a winning strategy.

\subparagraph{Memory.} 
\AP A ""finite-memory strategy""\footnote{It is common to define a memory structure as an automaton reading the edges of a game graph. This notion can be recovered by taking $k$ as the states of the automaton.}
is a strategy with vertices $V(G) \times k$ and projection $\pi(v,m)=v$, with the additional requirement that for every edge $(v,m) \re \eps (v',m')$, it holds that $m=m'$.
\AP We say that $k$ is the ""memory@@strat"" of the strategy, and numbers $1,\dots,k$ are called memory states.
Informally, the requirement above says that when reading an $\eps$-transition in the game, we are not allowed to change the memory state; this is called $\eps$-memory in~\cite{CO25LMCS} (to which we refer for more discussion), but since it is the main kind of memory in this paper, we will simply call it the memory.

\AP A finite-memory strategy is called ""chromatic@@graph"" if there is a map $\chi:k \times (\Sigma \cup \{\varepsilon\}) \to k$ such that for every edge $(v,m) \re c (v',m')$ in the strategy, it holds that $m'=\chi(m,c)$.
We say that $\chi$ is the chromatic update.
Note that necessarily, we have $\chi(m,\eps)=m$ for every "memory state" $m$.

\AP The (""chromatic@@mem"") ""memory"" of an "objective" $W$ is the minimal $k$ such that for every game with objective $W$, if Eve has a winning strategy, she has a winning ("chromatic@@mem") strategy with memory $\leq k$.

\subsection{Automata}

\AP A ""parity automaton"" $\A$ (or just automaton in the following) with ""index"" $d$ -- an even number -- and alphabet $\Sigma$, is a "pointed" $((\Sigmaeps) \times d)$-graph.
Vertices are called states, edges are called transitions and written $q \re{c:y} q'$, where $c \in \Sigmaeps$ and $y \in d$.
Elements in $d$ are called priorities.
Generally, we use the convention that even priorities are denoted with letter~$x$, whereas $y$ can be used to denote any priority.
\AP Transitions of the form $q \re{\eps:y} q'$ are called ""$\eps$-transitions""; note that they also carry priorities.

Infinite paths from the initial state $q_0$ are called runs. A run is accepting if the projection of its label on the second coordinate belongs to
\[
    \intro*\Parity_d = \{y_0y_1 \dots \in d^\omega \mid \liminf(y) \text{ is even}\}. \quad \text{ (Note the use of \emph{min}-parity.)}
\]
\AP The ""language"" $L(\A)$ of $\A$ is $\pi_\Sigma(L')$, where $L' \subseteq (\Sigmaeps)^\omega$ is the set of projections on the first coordinate of runs which are accepting.
We require that all these projections are infinite words; stated differently, there is no accepting run from $q_0$ labelled by a word in $(\Sigmaeps)^* \eps^\omega$.
\AP An "automaton" is ""deterministic"" if there are no "$\eps$-transitions" and for any state $q \in V(\A)$ and any letter $a \in \Sigma$ there is at most one transition $q \re{a:\ph} \ph$.
We say that an automaton is ""determinisable by pruning"" if one can make it deterministic by removing some transitions, without modifying its language.
\AP A language belongs to $\intro*\BCSigma$ if it is the language of a deterministic automaton, and it is ""$\omega$-regular"" if the automaton is moreover finite.


\AP We will often identify "pointed graphs" with "automata" of "index" $2$ whose transitions are all labelled with priority $0$.
Note that in this case, all runs are accepting.
This requires making sure that there is no accepting run labelled by a word from $\Sigma^* \eps^\omega$, which, up to assuming that all vertices are accessible from the initial one, amounts to saying that there is no infinite path of $\re \eps$.
We say that such a graph is ""well-founded"".

\subparagraph*{Blowups and $k$-automata.} 
\AP A ""$k$-blowup"" $\B$ of an "automaton" $\A$ is any automaton with $V(\B) \subseteq V(\A) \times \{1,\dots, k\}$, with initial state in $\{q_0\} \times k$ and such that
\begin{enumerate}
    \item for each transition $q \re{c:y} q'$ in $\A$ and each $m \in k$, there is a transition $(q,m) \re{c:y} (q',m')$ in $\B$ for some $m' \in k$, and
    \item all "$\eps$-transitions" are of the form $(q,m) \re {\eps:y} (q',m)$ for some $m \in k$.
\end{enumerate}
Note that the first item implies that $L(\A) \subseteq L(\B)$.
However there may be additional transitions (e.g. "$\eps$-transitions") that are not of the form described in the first item.

\AP A ""$k$-automaton"" is just an automaton whose states are a subset of $Q \times k$, for some set $Q$; for instance, "$k$-blowups" are "$k$-automata".
Equivalently, these are automata with a partition of their states in $k$ identified subsets.
\AP For a state $(q,m)$ in a "$k$-automaton", $m$ is called its ""memory state"".
\AP A "$k$-automaton" is called ""chromatic@@aut"" if there is a map $\chi:  \{1,\dots, k\} \times (\Sigma \cup \{\eps\}) \to  \{1,\dots, k\}$ such that for all transition $(q,m) \re{c:y} (q',m')$ it holds that $m'=\chi(m,c)$.

\subparagraph{Cascade products.} Let $\A$ be an "automaton" with alphabet $\Sigma$ and "index" $d$, and let $S$ be a "$d$-graph".
\AP We define their ""cascade product"" $\A \intro*\casc S$ to be the $(\Sigmaeps)$-graph with vertices $V(\A) \times V(S)$ and edges
\[
    (q,s) \re c (q',s') \quad \iff \quad \exists y, [q \re {c:y} q' \text{ and } s \re y s'].
\]
If $S$ is a "pointed graph" with intial vertex $s_0$, then $\A \casc S$ is "pointed" with initial vertex $(q_0,s_0)$.
It is easy to check that we then have the following lemma.

\begin{lemma}\label{lem:cascade_products}
    Let $\A$ be an automaton with index $d$ and $S$ be a "$d$-graph" satisfying $\Parity_d$.
    Then $\A \casc S$ is "well-founded" and "satisfies" $L(\A)$.
\end{lemma}

\subsection{$\eps$-completability and universal graphs}
We now introduce and discuss the key notion used in our main characterisation, which adapts the notion from~\cite{CO24Positional} from positionality to finite memory.

\subparagraph*{$k$-wise $\eps$-completability.}
\AP A "$k$-automaton" $\A$ with "index" $d$ is called ""$k$-wise $\eps$-complete"" if for each even $x \in d$, for each "memory state" $m$ and each ordered pair of states $(q,m),(q',m)$:
\[
    \text{either} \quad (q,m) \re{\eps:x} (q',m) \quad \tor \quad (q',m) \re{\eps:x+1} (q,m). 
\]

Intuitively, having an edge $(q,m) \re{\eps:x} (q',m)$ means that ``$(q,m)$ is much better than $(q',m)$'', as one may freely go from $(q,m)$ to $(q',m)$ and even see a good priority on the way.
Similarly, $(q',m) \re{\eps:x+1} (q,m)$ means that ``$(q',m)$ is not much worse than $(q,m)$''.

It is also useful for the intuition to apply the definition to both ordered pairs $((q,m),(q',m))$ and $((q',m),(q,m))$.
Since automata exclude accepting runs which are ultimately comprised of "$\eps$-transitions", we cannot have $(q,m) \rer{\eps:x} (q',m)$, and therefore "$\eps$-completability" rewrites as: for each $x$, each "memory state" $m$ and each unordered pair $(q,m),(q',m)$ of states,
\[
    \text{either} \quad (q,m) \xrightarrow[\eps :x+1]{\eps:x} (q',m) \quad \tor \quad (q,m) \rer{\eps:x+1} (q',m). 
\]
Hence an alternative, maybe more useful view (this is the point of view adopted in~\cite{CO24Positional}) is that, up to applying some adequate closure properties, 
 a "$k$-wise $\eps$-complete" automaton is endowed with the following structure: for each even priority $x$ and each "memory state" $m$, the states with "memory state" $m$ are totally preordered by the relation $\re {x+1}$, and the relation $\re x$ is the strict version of this preorder.
Moreover, for $x'>x$, the $x'$-preorder is a refinement of the $x$-preorder.

\AP A "$k$-automaton" $\A$ is called ""$k$-wise $\eps$-completable"" if one may add $\eps$-transitions to it so as to turn it into a "$k$-wise $\eps$-complete" automaton $\A^\eps$ satisfying $L(\A^\eps)=L(\A)$. In this case, we call $\A^\eps$ an ""$\eps$-completion"".
We simply say ""$\eps$-complete"" (resp. ``""$\eps$-completable""'') when $k$ is clear from the context.

\begin{example}
    Let $\Sigma = \{a, b,c\}$ and 
    \[W = \noOcc(b) \vee \finOften(aa) \vee \infOften(cc).\]
    We show a "$2$-wise $\eps$-complete" automaton recognising $W$ in Figure~\ref{fig:aut-eps-complete}. By Theorem~\ref{thm:main-charac}, the memory of $W$ is $\leq 2$ (and it is easy to see that this bound is tight).

    \begin{figure}[h]
        \centering
        \includegraphics{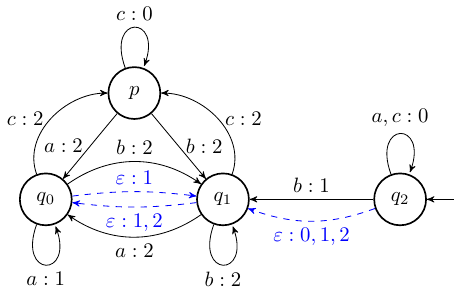}
        \caption{A 2-automaton recognising $W=\noOcc(b) \vee \finOften(aa) \vee \infOften(cc)$, where $p$ is assumed to have a different "memory state" than $q_0,q_1$ and $q_2$. It is "$2$-wise $\eps$-completable" by adding the indicated "$\eps$-transitions". A "completion" also contains the transitions $q_2\re{\eps:0,1,2} q_0$, as well as all transitions $q_i \re{\eps:3} q_{j}$ for $i>j$, and transitions $p \re{\eps:y} p$ for $y$ odd; these are omitted for ease of reading. However, it is not "chromatic@@aut" since reading $c$ may or may not switch the "memory state".}
        \label{fig:aut-eps-complete}
    \end{figure}
\end{example}

In the important special case of "well-founded" graphs viewed as automata, since all non-$\eps$-transitions are labelled with priority $0$, transitions $\re {\eps:0}$ and $\re{\eps:1}$ play the same role, and we simply write them $\re \eps$.
Therefore,~\ref{eq:1} can be seen as totality of the $\re \eps$ relation.

The following theorem, a key result in~\cite{CO25LMCS} where it is called the structuration lemma, will also be crucial to this work.
Recall that we see "well-founded" "pointed graphs" as "automata" with only $0$-transitions, and we apply the terms "$k$-blowup" and "$\eps$-completable" to them accordingly.

\begin{restatable}[{Adapted from Lemma 3.4 in \cite{CO25LMCS}}]{theorem}{structuration}
    \label{thm:structuration}
    Let $G$ be a "well-founded" "pointed graph" "satisfying" an "objective" $W$ which is assumed to have ("chromatic@@mem") "memory" $\leq k$ over games of size $\leq 2^{|G|}$.
    There is a ("chromatic@@mem") "$k$-blowup" $G'$ of $G$ which is "well-founded", "$k$-wise $\eps$-complete", and "satisfies" $W$.
\end{restatable}

For completeness, we give a proof of Theorem~\ref{thm:structuration} in Appendix~\ref{app:structuration}.

\subparagraph*{Universal graphs.}
\AP Given an objective $W$ and a cardinal $\kappa$, we say that a graph $U$ is ""$(\kappa,W)$-universal"" if for any "infinite tree" $T$ of cardinality $|V(T)|<\kappa$ "satisfying" $W$, there is a "morphism" $\phi:T \to U$ such that $\phi(t_0)$ satifies $W$ in $U$, where $t_0$ is the root of $T$.
We may now rephrase the main theorem of~\cite{CO25LMCS} in terms of "$\eps$-complete" "universal graphs".

\begin{theorem}[Theorem 3.1 in \cite{CO25LMCS}]\label{thm:universal_graphs}
Let $W$ be an objective.
Then $W$ has ("chromatic@@mem") "memory" $\leq k$ if and only if for every cardinal $\kappa$ there exists a "$(\kappa,W)$-universal" graph which is ("chromatic@@graph" and) "$k$-wise $\eps$-complete".
\end{theorem}

\AP We now give an explicit definition of a "$(\kappa,\Parity_{d})$-universal" graph $\intro*\Skappad$ which is "$\eps$-complete".
These ideas date back to the works of Streett and Emerson, who coined the name signatures~\cite{SE89}, and were made more explicit by Walukiewicz~\cite{Walukiewicz96}. 
Vertices are tuples of ordinals $<\kappa$, indexed by odd priorities in $d$ and ordered lexicographically (with the smaller indices being the most significant).
For a tuple $s$ and index $y$, we let $s_{< y}$ be the tuple obtained from $s$ by restricting to coordinates with index $< y$. 
Edges are given by
\[
    s \re y s' \quad \iff \quad \begin{cases}
        s_{< y} \geq s'_{< y} \text{ and $y$ is even; or}\\
        s_{\leq y} > s'_{\leq y} \text{ and $y$ is odd.}
    \end{cases}
\]
In particular, note that $s \re{d-1} s'$ if and only if $s > s'$.

\begin{lemma}[{\cite[Lemma~2.7]{CO24Arxiv}}]
The graph $\Skappad$ is $(\kappa,\Parity_d)$-universal.
\end{lemma}

\subparagraph*{Signature trees.}
We will work extensively with the graph $\Skappad$ defined above, and manipulations of its vertices which are tuples of ordinals indexed by odd priorities up to $d$.
\AP For a number $x$, we use $\intro*\xodd=\{1,3,\dots,x-1\}$ to denote the set of odd priorities $<x$.

For readability, we use subscripts to indicate which (odd) coordinates are concerned, for instance $s_{<x}$ will be our notation for tuples of ordinals $< \kappa$ indexed with odd priorities $<x$, and similarly for $s_{>x}$.
Therefore we often use $s_{<x}$ and $s_{>x}$ as two different variables (and not necessarily the projections of some given variable $s$).
Concatenation of tuples is written like for words, therefore $s_{<x} s_{>x}$ denotes a tuple indexed by all odd priorities (i.e. a vertex of $\Skappad$).

\AP By a slight abuse (since these do not correspond to "trees@@infinite" are defined above), we use the terminology ""signature trees"" to refer to subsets of $\kappa^{\dodd}=V(\Skappad)$.
Elements of the subsets should be thought of as leaves of the tree, while their (non-proper) prefixes correspond to nodes.
\AP More precisely, a ""node at level $x$"", where $x$ is an even priority from $d$, in a signature tree $T$ is a tuple $s_{< x} \in \kappa^{\xodd}$ such that there exists $s_{> x}$ satisfying $s_{< x} s_{> x} \in T$.
In particular, elements of $T$ are "nodes at level $d$" (i.e. leaves).
\AP The ""subtree rooted at"" a node $s_{<x}$ of level $x$ is defined to be $\{s_{>x} \mid s_{<x}s_{>x} \in T\}$.

The children of a node $s_{<x}$ are the nodes of the form $s_{<x}s_{x+1}$.
\AP The branching of a node is its number of children, and we say that a "signature tree" has ""branching"" $b$ if every node has branching exactly $b$.
Notably, $\kappa^{\dodd}$ has "branching" $\kappa$.

%% file: mainCharacterisation.tex
We state our main characterisation theorem and its decidability consequences for "$\omega$-regular" languages.
We assume that the alphabet $\Sigma$ is countable, therefore "automata" can also be taken with countable sets of states.

\begin{theorem}[Main characterisation]\label{thm:main-charac}
Let $W$ be a $\BCSigma$ "objective" and let $k \in \N$.
The following are equivalent:
\begin{enumerate}[(i.)]
    \item\label{item:memory-small-games} $W$ has "memory" $\leq k$ (resp. chromatic "memory" $\leq k$) on "games" of size $\leq 2^{2^{\aleph_0}}$.
    \item\label{item:existence-automata} For any "automaton" $\A$ recognising $W$, there is a ("chromatic@@mem") "$k$-blowup" $\B$ of $\A$ which is "$k$-wise $\eps$-complete" and recognises $W$.
    \item\label{item:existence-det-automata} There is a "deterministic" ("chromatic@@mem") "$k$-automaton" $\A$ which is "$k$-wise $\eps$-completable" and recognises $W$. If $W$ is recognised by a "deterministic" "automaton" of size $n$, then $\A$ can be taken of size $kn$.
    \item\label{item:existence-universal-graph} For every cardinal $\kappa$, there is a ("chromatic@@mem")  $(\kappa,W)$-"universal graph" which is "well-founded" and "$k$-wise $\eps$-complete".
    \item\label{item:memory-arbitrary-games} $W$ has ("chromatic@@mem") "memory" $\leq k$ on arbitrary "games".
\end{enumerate}

Moreover in the case where $W$ is "$\omega$-regular" and recognised by an "automaton" with $n$ states and "index" $d$, this is also equivalent to:
\begin{enumerate}[(i'.)]
    \item\label{item:memory-finite-games} $W$ has "memory" $\leq k$ (resp. chromatic "memory" $\leq k$) on "games" of size $\leq f(k, |\Sigma|, n, d)$, where $f$ is some triply exponential function explicited in Proposition~\ref{prop:existenceEpsComplete-finitary} below.
\end{enumerate}
\end{theorem}

This immediately gives the finite-to-infinite lift for both memory and chromatic memory.

\begin{corollary}[Finite-to-infinite lift]\label{cor:finite-to-infinite-lift}
    An "$\omega$-regular" "objective" has "memory" (resp. chromatic memory) $\leq k$ over finite "games" if and only if it has "memory" (resp. chromatic memory) $\leq k$ over all games.
\end{corollary}

For "$\omega$-regular" $W$ given by a "deterministic" "automaton" $\B$ of size $n$, this also allows to compute the ("chromatic@@mem") "memory" in $\NP$. First, we note that the "memory" of $L(\B)$ is at most $n$, as the "automaton" itself can serve as a ("chromatic@@mem") memory. Therefore, we can guess $k\leq n$, a "deterministic" "automaton" $\A$ of size $\leq kn$ and a ("chromatic@@mem") $k$-wise $\eps$-"completion" $\A^\eps$, and check if  $L(\B) \subseteq L(\A)$ and if $L(\A^\eps) \subseteq L(\B)$, which can be done in polynomial time, since $\A$ and $\B$ are "deterministic"~\cite{ClarkeDK93Unified}.
Prior to our work, computing the "memory" was not known to be decidable, and computing the "chromatic memory" was only known\footnote{In fact, Kopczyński proved that the "chromatic memory" \emph{over finite games} is computable. By Corollary~\ref{cor:finite-to-infinite-lift}, this coincides with the "chromatic memory" over arbitrary games.} do be doable in exponential time~\cite{Kop08Thesis}.

\begin{corollary}[Decidability in $\NP$]\label{thm:NP-computation-memory}
    Given an integer $k$ and a "deterministic" "automaton" $\A$, the problem of deciding if $L(\A)$ has ("chromatic@@mem") "memory" $\leq k$ belongs to $\NP$.
\end{corollary}

A third important consequence of Theorem~\ref{thm:main-charac} is a closure under union -- a strong form of the "generalised Kopczyński conjecture" -- stated below. The proof of this result will be the object of Section~\ref{sec:union}.
An objective is ""prefix-increasing"" if for all $a \in \Sigma$ and $w \in \Sigma^\omega$, it holds that if $w \in W$ then $aw \in W$.
\AP It is ""prefix-independent"" if the converse also holds, that is, $w \in W$ if and only if $aw \in W$.

\begin{restatable}[Union has bounded memory]{theorem}{thmUnion}\label{thm:union}
    Let $W_1, W_2 \subseteq \SS^\oo$ be two $\BCSigma$ objectives over the same alphabet, such that $W_{2}$ is "prefix-increasing".
    Assume that $W_{1}$ has memory $\leq k_1$ and $W_{2}$ has memory $\leq k_2$.
    Then $W_1 \cup W_2$ has memory $\leq k_1 k_2$.
\end{restatable}

Our main technical contribution are the implications from (\ref{item:memory-small-games}) to (\ref{item:existence-automata}) and from (i') to (\ref{item:existence-automata}) in the "$\omega$-regular" case, which are the objects of Sections~\ref{sec:existence_automata} and~\ref{sec:finitary}.
We proceed in Section~\ref{sec:existence-det-automata} to show that (\ref{item:existence-automata}) implies (\ref{item:existence-det-automata}) which is straightforward.
The implication (\ref{item:existence-det-automata}) $\implies$ (\ref{item:existence-universal-graph}) is adapted from~\cite{CO24Positional} and presented in Section~\ref{sec:existence-universal-graphs}.
Finally, the implication (\ref{item:existence-universal-graph}) $\implies$ (\ref{item:memory-arbitrary-games}) is the result of~\cite{CO25LMCS} (Theorem~\ref{thm:universal_graphs}), and the remaining one is trivial.

\subsection{Existence of $k$-wise $\eps$-complete automata: infinitary proof}\label{sec:existence_automata}

We start with the more challenging and innovative implication: how to obtain a "$k$-wise $\eps$-complete" automaton given an "objective" in $\BCSigma$ with "memory" $k$ (that is, (\ref{item:memory-small-games}) $\implies$ (\ref{item:existence-automata})).
The finitary version (that is, (i') $\implies$ (\ref{item:existence-automata})) follows the same lines but requires some additional insights, it is the object of the next section (Section~\ref{sec:finitary}).
The proof we propose yields a slightly stronger result than the implication (\ref{item:memory-small-games}) $\implies$ (\ref{item:existence-automata}): we can obtain a "$k$-wise $\eps$-complete" automaton from any automaton recognising a language included in $W$, as formalised below.

\begin{restatable}{proposition}{existenceEpsComplete}
    \label{prop:existenceEpsComplete}
    Let $W$ be an objective with ("chromatic@@mem") memory $\leq k$ on "games" of size $\leq 2^{2^{\aleph_0}}$ and let $\A$ be an automaton such that $L(\A) \subseteq W$.
    Then there is a ("chromatic@@mem") "$k$-blowup" $\B$ of $\A$ which is "$k$-wise $\eps$-complete" and such that $L(\A) \subseteq L(\B) \subseteq W$.
\end{restatable}

To prove (\ref{item:memory-small-games}) $\implies$ (\ref{item:existence-automata}), we apply the proposition to an automaton $\A$ recognising $W$, which ensures that the obtained automaton $\B$ also recognises $W$.
We start with a detailed proof overview (Section~\ref{sec:overview}), then move on to the formal proof (Sections~\ref{sec:combinatorial-infinitary},~\ref{sec:def-of-B} and~\ref{sec:correctness-of-B}).

\subsubsection{Proof overview}\label{sec:overview}

Let $\kappa=2^{\aleph_0}$.
We assume that $W$ has "memory" $\leq k$ on "games" of size $\leq 2^\kappa$; we discuss the chromatic case at the end of the section.
Let $\A$ be an "automaton" such that $L(\A) \subseteq W$; we aim to construct a "$k$-blowup" $\B$ of $\A$ which is $\eps$-complete and satisfies $L(\A) \subseteq L(\B) \subseteq W$.
We let $\Sgraph$ denote $\Skappad$, the $(\kappa,\Parity_{d})$-"universal graph" defined in the preliminaries.

We consider the "cascade product" $\A \casc S$; this is a $(\Sigma \cup \{\eps\})$-graph which intuitively encodes all possible accepting behaviours in $\A$.
Then we apply the structuration result (Theorem~\ref{thm:structuration}) to $\A \casc S$ which yields a "$k$-blowup" $G$ of $\A \casc S$ which is "well-founded" and "$k$-wise $\eps$-complete" (as a graph).
Stated differently, up to blowing the "graph" $\A \casc S$ into $k$ copies, we have been able to endow it with many $\eps$-transitions, so that over each copy, $\re \eps$ defines a well-order.
Note that the states of $G$ are of the form $(q,m,s)$, with $q\in V(\A)$, $m \in k$ and $s\in V(S)$.

The states of $\B$ will be $V(\B)=V(\A) \times k$.
The challenge lies in defining the transitions in $\B$, based on those of $G$.

Given a state $(q,m)\in V(\B)$ and a transition $q \re {c:y} q'$ in $\A$, where $c \in \Sigma \cup \{\eps\}$, by applying the definitions we get transitions of the form $(q,m,s) \re c (q',m',s')$ in $G$, for different values of $m'$, whenever $s \re y s'$ in $\Sgraph$.
We will therefore define transition $(q,m) \re{c:y} (q',m')$ in $\B$ if $m'$ matches suitably many transitions as above; for now, we postpone the precise definition.

We should then verify that the obtained "automaton" $\B$:
\begin{itemize}
    \item is a "$k$-blowup" of $\A$,
    \item is $\eps$-complete, and
    \item recognises a subset of $W$.
\end{itemize}
The first two items above state that $\B$ should have many transitions: at least those inherited from $\A$, and in addition a number of $\eps$-transitions.
This creates a tension with the third item, which states that even with all these added transitions, the "automaton" $\B$ should not accept too many words.

Let us focus on the third item for now, which will lead to a correct definition for $\B$.
Take an accepting run
\[
    (q_0,m_0) \re{c_0:y_0} (q_1,m_1) \re{c_1,y_1} \dots
\]
in $\B$, where $x = \liminf_i y_i$ is even. For the sake of simplicity, assume that all $y_i$'s are $\geq x$.
We should show that its labelling $w=c_0c_1\dots$ belongs to $W$.
To this end, we will decorate the run with labels $s_0,s_1,\dots \in S$ so that
\[
    (q_0,m_0,s_0) \re{c_0} (q_1,m_1,s_1) \re{c_1} \dots
\]
defines a path in $G$, which concludes since $G$ "satisfies" $W$.

Recall that the elements of $\Sgraph$ are tuples of ordinals $<\kappa$ indexed by odd priorities up to $d$, and that we use $s_{<x}$ (resp. $s_{>x}$) to refer to a tuple indexed by odd priorities up to $x-1$ (resp. from $x+1$).
To construct the $s_i$'s, we fix a well chosen prefix $s_{<x}\in \kappa^{\xodd}$ which will be constant, and proceed as follows.

\begin{enumerate}[(a.)]
    \item If $y_i=x$, then we set $s_i=s_{<x} s_{>x}$, for some $s_{>x}$ which depends only on $c_i$.
    \item If $y_i>x$, then we set $s_i=s_{<x} s_{>x}$, for some $s_{>x}$ which depends on $c_i$ as well as $s_{i+1}$.
\end{enumerate}

At this stage the reader may be worried that the backward induction underlying the above definition is not "well-founded"; however, since the first case occurs infinitely often, the backward induction from the second item is only performed over finite blocks (see also Figure~\ref{fig:backtrack} in Section~\ref{sec:correctness-of-B}).

This leads to the following definition for $\B$, where $x$ is an even priority:
\[
    \begin{array}{rcl}
     (q,m) \re{c:x} (q',m') \tin \B  &\iff&  \exists s_{<x} \exists s_{>x} \forall s'_{>x}, (q,m,s_{<x}s_{>x}) \re{c} (q',m',s_{<x} s'_{>x}) \tin G, \\
     (q,m) \re{c:x+1} (q',m') \tin \B  &\iff&  \exists s_{<x} \forall s'_{>x} \exists s_{>x}, (q,m,s_{<x}s_{>x}) \re{c} (q',m',s_{<x} s'_{>x}) \tin G.
    \end{array}
\]

The first line corresponds to point (a.) above, where $s_{>x}$ can be chosen independently of $s'_{>x}$, whereas the second line corresponds to point (b.), since the choice of $s_{>x}$ is conditioned on the value of $s'_{>x}$.
(For priorities $>x+1$, we may apply either the first or second line, depending on the parity, to get the required conclusion.)

The remaining issue is that the choice of the fixed common prefix $s_{<x}$ should be made uniformly, regardless of the transition.
This is achieved thanks to an adequate extraction lemma (which extends the pigeonhole principle to the case at hands), which finds a large enough subset $T$ of $\kappa^{\dodd}$, so that transitions $(q,m,s)\re{}(q',m',s')$ are similar for different choices of $s,s' \in T$.
This ensures that $s_{<x}$ can be chosen uniformly.

There remains to verify that $\B$ is indeed a "$k$-blowup" of $\A$ and that it is $\eps$-complete, which will follow easily from the definitions and $\eps$-completeness of $G$ (this is because, after removing ``$\exists s_{<x}$'' from the definition above, the second line resemble the negation of the first).

For the chromatic case, the proof is exactly the same, we should simply check that if $G$ is chromatic (which is guaranteed by Theorem~\ref{thm:structuration}), then so is the obtained "automaton" $\B$.

We now present the full details of the proof, starting with the extraction lemma.

\subsubsection{A combinatorial lemma: Extracting homogeneous subtrees}\label{sec:combinatorial-infinitary}

As an important part of our proof, we will take the "graph" $\Sgraph$, whose set of vertices is $\kappa^{\dodd}$, and extract from it a large enough subgraph which is homogeneous.

\AP We say that a tree $T$ is ""everywhere cofinal"" if for each node $s_{< x}$, the "subtree rooted at" $x$ is cofinal in $\kappa^{(d-x)_\odd}$.
\AP An ""inner labelling"" of a tree $T$ by $L$ is a map $\lambda$ assigning a label in $L$ to every node in $T$.
\AP We say that an "inner labelling" is ""constant per level"" if for every $x\in\dodd$, $\lambda$ is constant over nodes of level $x$ in $T$.

We are now ready to state the extraction lemma.
Recall that $\kappa=2^{\aleph_0}$.

\begin{lemma}\label{lem:everywhere_cofinal}
Let $\lambda$ be an "inner labelling" of $\kappa^{\dodd}$ by $L$, where $L$ is countable.
There is an "everywhere cofinal" tree $T$ such that at every level $x$, $\lambda|_T$ is "constant per level".
\end{lemma}

\begin{proof}
    We prove the lemma by induction on $d$.
    For $d=0$ there is nothing to prove since $\dodd$ is empty; let $d\geq 1$ and assume the result known for $d-2$.
    For each node $s_{<2}$ at level $2$, apply the induction hypothesis on the subtree of $\kappa^{\dodd}$ rooted at $s_{<2}$,
    which gives an everywhere cofinal tree $T_{s_{<2}}$ such that $\lambda$ is "constant per level" over $T_{s_{<2}}$.
    Let $\ell^{s_{<2}}_2,\dots,\ell^{s_{<2}}_{d}$ denote the constant values of $\lambda$ on the corresponding levels of $T_{s_{<2}}$, and define a new auxiliary labelling of the nodes $s$ at level $2$ of $\Sgraph$ by the tuple $\ell = (\lambda(s),\ell^{s_{<2}}_2,\dots,\ell^{s_{<2}}_{d})$.

    Now since there are at most countably-many new labellings, and there are $\kappa=2^{\aleph_0}$ nodes at level $2$, there is one of the auxiliary labels $\ell$ such that cofinaly-many nodes have this new label.
    We conclude by taking $T$ to be the union of $\{s\} \times T'_s$, where $s$ ranges over nodes at level $2$ with the new labelling $\ell$, and $T'_s$ are the corresponding everywhere cofinal trees.
\end{proof}

We are now ready to formalise the definition of $\B$.

\subsubsection{Definition of $\B$}\label{sec:def-of-B}

Without loss of generality, we assume that $\A$ contains transition $q \re{\eps:d-1} q$ for every state $q$ (these can be added without affecting the language of $\A$).
Consider the "cascade product" $\A \casc S$.
Note that thanks to the assumption above and the definition of $\Sgraph$, we have transitions $(q,s) \re{\eps} (q',s)$ in $\A \casc S$ whenever $s>s'$.

By Lemma~\ref{lem:cascade_products}, $\A \casc S$ "satisfies" $W$ and is "well-founded".
Moreover it has size $\leq \kappa$, so by our assumption on $W$, we may apply Theorem~\ref{thm:structuration}.
This yields a "$k$-blowup" $G$ of ${\A \casc S}$ which is "$k$-wise $\eps$-complete".
Let us write $V(G) = V(\A) \times k \times V(S)$.
We close $G$ by transitivity, meaning that we add transitions $(q,m,s) \re c (q',m',s')$, for $c \in \Sigmaeps$, whenever $(q,m,s) \re{\eps^* c \eps^*} (q',m',s')$; the obtained "graph" $\overline{G}$ still "satisfies" $W$.

Here comes the important definition: say that $(q,m)$ strongly $c$-dominates $(q',m')$ at node $s_{< x}$ if
\[
    \exists s_{> x} \forall s'_{> x} \qquad (q,m,s_{< x} s_{> x}) \re c (q',m',s_{< x} s'_{> x}) \tin \overline{G},
\]
and that $(q,m)$ weakly $c$-dominates $(q',m')$ at $s_{< x}$ if
\[
    \forall s'_{> x} \exists s_{> x} \qquad (q,m,s_{< x} s_{> x}) \re c (q',m',s_{< x} s'_{> x}) \tin \overline{G},
\]
where $c \in \Sigmaeps$.
Note that strong domination implies weak domination.
The type of a node $s_{< x}$ is the information, for each $q,q',m,m'$ and $c$, of whether $(q,m)$ strongly or weakly (or not at all) $c$-dominates $(q',m')$.
This gives finitely many possibilities for fixed $q,q',m,m'$ and $c$, and therefore there are in total a countable number of possible types.
Thus Lemma~\ref{lem:everywhere_cofinal} yields a tree $T \subseteq \kappa^{\alpha}$ which is everywhere cofinal and such that for all $x$, nodes at level $x$ in $T$ all have the same type $t_x$.

We are now ready to define $\B$.
We put $V(\B) = V(\A) \times k$, and for each even $x \in d$ and $c \in \Sigmaeps$, we define transitions by
\[
    \begin{array}{lcrl}
    (q,m) &\re{c:x}& (q',m') & \tif (q,m) \text{ strongly $c$-dominates } (q',m') \tin t_x, \\
    (q,m) &\re{c:x+1}& (q',m') & \tif (q,m) \text{ weakly $c$-dominates } (q',m')\tin t_x.
    \end{array}
\]
Here is the main lemma, which proves the direct implication in Theorem~\ref{thm:main-charac}.

\begin{lemma}\label{lem:main_lemma_charac}
"Automaton" $\B$ is a "$k$-blowup" of $\A$, it is "$k$-wise $\eps$-complete" and $L(\B) \subseteq W$.
\end{lemma}
The remainder of the section is devoted to proving Lemma~\ref{lem:main_lemma_charac}.

\subsubsection{Correctness of $\B$: Proof of Lemma~\ref{lem:main_lemma_charac}}\label{sec:correctness-of-B}

There are a few things to show.
The interesting argument is the one that shows that $L(\B) \subseteq W$ (Lemma~\ref{lem:language-containement} below).
We should also prove there is no accepting run over words in $\Sigma^*\eps^\omega$, which will be done below as part of Lemma~\ref{lem:language-containement}.

\subparagraph*{$\B$ is a $k$-blowup of $\A$.}

We should prove the following.

\begin{claim}\label{claim:blowup-infinitary}
    For all transitions $q \re{c:y} q'$ in $\A$, and any $m \in k$ there is some $m' \in k$ such that $(q,m) \re{c:y} (q',m')$ in $\B$.
\end{claim}

\begin{claimproof}
Let $q \re{c:y} q'$ be a transition in $\A$ and let $m \in k$.
Since $G$ is a "$k$-blowup" of $\A \times S$, for all edges $s\re y s'$ in $\Sgraph$ there is $m' \in k$ such that $(q,m,s) \re c (q',m',s') \tin G$.
Although both proofs are similar, we distinguish two cases.
\begin{itemize}
    \item If $y=x$ is even.
    We prove that for all nodes $s_{< x}$ at level $x$, $(q,m)$ strongly $c$-dominates $(q',m')$ for some $m'$.
    Therefore the same is true in $t_x$ which implies the wanted result.
    We let $s_{> x}=0_{> x}$, the zero sequence in $\kappa^{(d-x)_\odd}$.
    Now for all $s'_{>x} \in \kappa^{(d-x)_\odd}$, it holds that $s_{< x} s_{> x} = s_{< x} 0_{> x} \re x s_{< x} s'_{>x}$ in $\Sgraph$, so there is $m'$ such that $(q,m,s_{< x} s_{>x}) \re c (q,m',s_{< x} s'_{>x})$ in $G$ and thus also in $\overline G$; in this case say that $m'$ is good for $s'_{> x}$.
    
    Now we claim that if $m'$ is good for $\tilde s'_{> x} \geq s'_{> x}$, then it is also good for $s'_{> x}$.
    Indeed, as observed at the beginning of the section, we have $(q',s_{<x}\tilde s'_{> x}) \re \eps (q',s_{<x}s'_{> x})$ in $\A \times S$ therefore since $\eps$-transitions preserve the "memory state" in $G$ (by definition of a blowup) we have $(q',m',s_{<x} \tilde s'_{> x}) \re \eps (q',m',s_{<x} s'_{>x})$ in $G$ thus $(q,m,s_{<x}s_{> x}) \re c (q',m',s_{<x}s'_{> x})$ in $\overline{G}$.

    Therefore the sets $M'_{s'_{> x}}$ of $m'$ which are good for $s'_{> x}$ form a decreasing chain of non-empty subsets of $k$ and thus their intersection is non-empty: there is some $m'$ which is good for all $s'_{> x}$, as required.

    \item If $y=x+1$ is odd. 
    We now prove that for all nodes $s_{< x}$ at level $x$, $(q,m)$ weakly $c$-dominates $(q',m')$ for some $m'$.
    Let $s'_{> x} \in \kappa^{(d-x)_\odd}$.
    Then for any $s_{> x}$ such that $s_x > s'_x$, it holds that $s_{< x} s_{> x} \re {x+1} s'_{< x} s_{>x}$ in $\Sgraph$, so there is some $m'$ such that $(q,m,s_{< x} s_{>x}) \re c (q,m',s_{< x} s'_{>x})$ in $G$.
    Hence there is some $m'$ such that, for cofinitely many $s_{> x}$, $(q,m,s_{< x} s_{>x}) \re c (q,m',s_{< x} s'_{>x})$ is an edge in $G$ and thus also in $\overline G$; say that such an $m'$ is good for $s'_{> x}$.

    Now, we claim that if $m'$ is good for $\tilde s'_{> x} \geq s'_{> x}$, then it is also good for $s'_{> x}$.
    Indeed, as in the first case, we have $(q',m',\tilde s_{<x}s'_{> x}) \re \eps (q',m',s_{<x}s'_{> x})$ in $G$ thus for cofinitely many $s_{> x}$ we have $(q,m,s_{< x}s_{>x}) \re c (q',m',s_{<x} s'_{> x})$ in $\overline{G}$.

    We conclude just as above. \claimqedhere 
\end{itemize}
\end{claimproof}

\subparagraph*{$\B$ is "$k$-wise $\eps$-complete".}

We should prove the following claim.

\begin{claim}
    For every even $x$, "memory state" $m$ and states $q,q'$, either $(q,m) \re{\eps:x} (q',m)$ or $(q',m) \re{\eps:x+1} (q,m)$.
\end{claim}

\begin{claimproof}
    Assume that $(q,m) \re{\eps:x+1} (q',m)$ is not a transition in $\B$.
    Consider a node $s_{<x}$ at level $x$ in $T$: it has type $t_x$ and thus $(q,m)$ does not weakly $\eps$-dominate $(q',m)$ at $s_{<x}$.
    This rewrites as
    \[
        \exists s'_{> x} \forall s_{> x} \qquad (q,m,s_{< x} s_{>x}) \re \eps (q',m,s_{< x} s'_{>x}) \text{ is not an edge in } \overline G.
    \]
    Now since $\overline G$ is $\eps$-complete, we get 
    \[
        \exists s'_{> x} \forall s_{> x} \qquad (q',m,s_{< x} s'_{>x}) \re \eps (q,m,s_{< x} s_{>x}) \tin \overline G,
    \]
    therefore $(q',m)$ strongly $\eps$-dominates $(q,m)$ at $s_{<x}$, which concludes.
\end{claimproof}

\subparagraph*{$\B$ recognises a subset of $W$.}

We now turn to the more involved part.
We start by proving the following two technical lemmas.

\begin{lemma}\label{lem:technical1}
Assume that $(q,m) \re{c:y} (q',m') \tin \B$ for some $y$.
There is a map $f:T \to T$ such that for all $s \in T$, 
\begin{itemize}
    \item there is an edge $(q,m,f(s)) \re c (q',m',s)$ in $\overline G$; and
    \item it holds that $f(s)_{< y} = s_{< y}$.
\end{itemize}
\end{lemma}

\begin{proof}
Let $x+1$ be the smallest odd priority $\geq y$.
Since strong domination implies weak domination, and $(q,m) \re{c:y} (q',m')$, we have in any case that $(q,m)$ weakly $c$-dominates $(q',m')$ in $t_{x}$.
This means that for any $s'=s'_{< x}s'_{> x} \in T$, there exists $s_{> x}$ such that $(q,m,s'_{< x}s_{> x}) \re c (q',m',s'_{< x}s'_{> x}) \tin \overline G$.
Now since $T$ is everywhere cofinal, there exists $\tilde s_{> x} \geq s_{> x}$ such that $s'_{<x}\tilde s_{> x} \in T$, and we let $f(s')=s'_{<x}\tilde s_{> x}$.
Clearly $f(s')_{< y} = s'_{< x} = s'_{< y}$, and also, we have $(q,m,s'_{<x}\tilde s_{> x} ) \re \eps (q,m,s'_{<x}s_{> x})\tin G$ so the result follows from $\eps$-transitivity.
\end{proof}

\begin{lemma}\label{lem:technical2}
    Assume that $(q,m) \re{c:x} (q',m') \tin \B$ for some even $x$, and let $s_{< x}$ be a "node at level $x$" in $T$.
    There is $s_{> x} \in \kappa^{(d-x)_\odd}$ such that $s_{< x}s_{> x} \in T$ and for any $s'_{> x} \in \kappa^{(d-x)_\odd}$, $(q,m,s_{< x}s_{> x}) \re c (q',m',s_{<x}s'_{> x})$ in $\overline G$.
\end{lemma}

\begin{proof}
    From the definition of strong domination, there is $\tilde s_{>x}$ such that for all $s'_{>x}$, it holds that $(q,m,s_{<x} \tilde s_{>x}) \re c (q',m',s_{<x} s'_{>x})$ in $\overline G$. By everywhere confinality of $T$, there is $s_{>x} 
    \geq \tilde s_{>x}$ such that $s_{<x} s_{>x} \in T$.
    We conclude (as previously) using $\eps$-transitivity.
\end{proof}

We are now ready for the main argument.

\begin{lemma}\label{lem:language-containement}
    The language of $\B$ is contained in $W$.
    Moreover, there is no accepting run labelled by words in $\Sigma^* \eps^\omega$.
\end{lemma}

\begin{proof}
    Take an accepting run
    \[
        (q_0,m_0) \re{c_0:y_0} (q_1,m_1) \re{c_1:y_1} \dots
    \]
    in $\B$.
    Let $x=\liminf_i y_i$ (it is even since the run is accepting), and let $i_0$ be such that $y_i \geq x$ for $i \geq i_0$.

    As explained in the proof overview, our goal will be to endow each $(q_i,m_i)$ with some $s_i \in T$ such that for all $i$, $(q_i,m_i,s_i) \re{c_i} (q_{i+1},m_{i+1},s_{i+1}) \tin \overline G$.
    This implies the result since $\overline G$ "satisfies" $W$, and since it does not have paths labelled by words in $\Sigma^* \eps^\omega$ by well-foundedness.
    We pick an arbitrary node $s_{< x}$ at level $x$ in $T$ and proceed as follows:
    \begin{itemize}
        \item for each $i\geq i_0$ such that $y_i = x$, we let $s_{> x}$ be obtained from Lemma~\ref{lem:technical2} and set $s_i = s_{< x} s_{> x}$;
        \item for any other $i$, we proceed by backwards induction within each block (see Figure~\ref{fig:backtrack}) and let $s_i=f(s_{i+1})$, where $f$ is obtained by applying Lemma~\ref{lem:technical1} to transition $(q_i,m_i) \re{c_i:y_i} (q_{i+1},m_{i+1})$.
    \end{itemize}

    \begin{figure}[h]
        \begin{center}
            \includegraphics[width=0.9\linewidth]{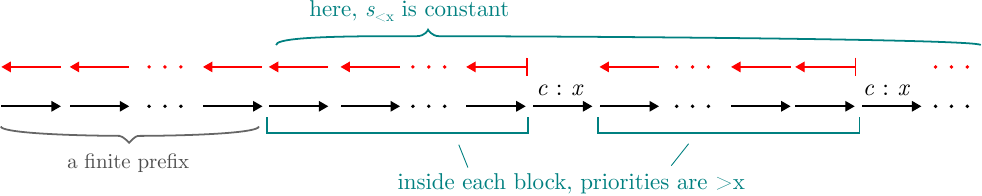}
        \end{center}
        \caption{The run in $\B$ (in black), and the order in which the $s_i$'s are computed (in red).}\label{fig:backtrack}
    \end{figure}

    For $i$'s as in the second item, it follows from Lemma~\ref{lem:technical1} that $(q_i,m_i,s_i) \re{c_i} (q_{i+1},m_{i+1},s_{i+1})$, and moreover, assuming $i \geq i_0$, that $(s_i)_{< x} = (s_{i+1})_{< x}$.
    Thus for all $i \geq i_0$ we have $(s_i)_{< x} = s_{< x}$ hence we also have, by Lemma~\ref{lem:technical2}, that $(q_i,m_i,s_i) \re{c_i} (q_{i+1},m_{i+1},s_{i+1}) \tin \overline G$ for $i$'s as in the first item. 
\end{proof}

\subsection{Existence of $\eps$-completable automata: finitary proof}\label{sec:finitary}

\input{finite-to-infinite}

\subsection{Existence of deterministic $\eps$-completable automata}\label{sec:existence-det-automata}

We now prove the implication from (\ref{item:existence-automata}) to (\ref{item:existence-det-automata}) in Theorem~\ref{thm:main-charac}.
Take $\A$ to be a deterministic "automaton" recognising $W$, and let $\B^\eps$ be the obtained "$k$-blowup" of $\A$ which is (chromatically) "$k$-wise $\eps$-complete" and recognises $W$.
Then let $\B$ be obtained from $\B^\eps$ by only keeping, for each state $(q,m) \in V(\B)$ and each transition $q \re{a:y} q'$ in $\A$, a single transition of the form $(q,m) \re {a:y} (q',m')$ chosen arbitrarily.
Note that $\B$ is a "$k$-blowup" of $\A$, so we have: $L(\A) \subseteq L(\B) \subseteq L(\B^\eps)$.
We conclude that $\B$ is a "deterministic" (chromatically) "$k$-wise $\eps$-completable" "automaton" recognising $W$, as required.

\subsection{From deterministic $\eps$-completable automata to universal graphs}\label{sec:existence-universal-graphs}

We now prove the implication from (\ref{item:existence-det-automata}) to (\ref{item:existence-universal-graph}) in Theorem~\ref{thm:main-charac}.
This result was already proved in~\cite[Prop.~5.30]{CO24Arxiv} for the case of $k=1$; extending to greater values for $k$ presents no difficulty.

Let $\B$ be a "deterministic"\footnote{For the purpose of this proof, history-determinism would be sufficient (we refer to~\cite{BL23SurveyHD} for the definition and context on history-determinism).} ("chromatic@@mem") "$k$-wise $\eps$-completable" "automaton" recognising $W$ and $\B^\eps$ be an "$\eps$-completion".
Fix a cardinal $\kappa$ and let $\Sgraph$ denote $S_d^\kappa$.
Consider the "cascade product" $U = \B^\eps \casc S$.
We claim that $U$ is ("chromatic@@mem") "$k$-wise $\eps$-complete" and $(\kappa,W)$-"universal".
Universality follows from the facts that $\B$ is "deterministic" and $\Sgraph$ is $(\kappa,\Parity_d)$-"universal".

\begin{claim}
    The "graph" $U = \B^\eps \casc S$ is $(\kappa,W)$-"universal".
\end{claim}

\begin{claimproof}
    Take a tree $T$ of size $<\kappa$ which "satisfies" $W$.
    We should show that $T \to U$.
    Since $\B$ is "deterministic", we can define a labelling $\rho\colon V(T) \to V(\B)$ by mapping $t_0$ to $q_0$ and $t \mapsto q$ if the run of $\B$ on the finite word labelling the path $t_0\re{}t$ ends in $q$. Then, any infinite path from $t_0$ in $T$ is mapped to a run in $\B$ that is accepting (since $T$ "satisfies" $W$).
    Therefore the tree $T'$ obtained by taking $T$ and replacing edge-labels by the priorities appearing in their $\rho$-images "satisfies" $\Parity_d$ and has size $<\kappa$, so there is a "morphism" $\mu:T \to S$.
    It is a direct check that $(\rho,\mu) \colon V(T) \to V(\B) \times V(S) = V(U)$ indeed defines a "morphism".
\end{claimproof}

Showing that $U$ is "$k$-wise $\eps$-complete" is slightly trickier.

\begin{claim}
    The "graph" $U = \B^\eps \casc S$ is "well-founded" and ("chromatic@@mem") "$k$-wise $\eps$-complete".
\end{claim}

\begin{claimproof}
    Well-foundedness of $U$ follows directly from Lemma~\ref{lem:cascade_products}.
    Let us write $B_1,\dots,B_k$ for the $k$ parts of $\B^\eps$; the $k$ parts of $U$ will be $B_1 \times V(S),\dots, B_k \times V(S)$.
    If applicable, chromaticity is a direct check.
    Let $(b,s),(b',s') \in V(U)$ be in the same part, i.e.~$b,b' \in B_i$ for some $i$.
    Let $x_0$ be the minimal even priority such that 
    $b \re{\eps:x_0} b'$ in $\B^\eps$ 
    (if such an $x$ does not exist then $x_0=d+2$).
    Then let $y_0$ be the minimal odd priority such that $s'_{\leq y_0} > s_{\leq y_0}$ (as previously, if $s=s'$ then we let $y_0=d+1$).
    We distinguish two cases.
    \begin{enumerate}[(1)]
        \item If $x_0 < y_0$. Then we have $b \re{\eps:x_0} b'$ in $\B^\eps$ and $s_{<x_0} = s'_{<x_0}$ which gives $s \re{x_0} s'$ in $\Sgraph$.
        Thus we get $(b,s) \re{\eps} (b',s')$ in $U$.
        \item If $y_0 < x_0$. Then $s'_{\leq y_0} > s_{\leq y_0}$, which gives $s' \re{y_0} s$ in $\Sgraph$.
        Since $\B^\eps$ is "$k$-wise $\eps$-complete" and by definition of $x_0$, we also have $b' \re{\eps:y_0} b$ in $\B^\eps$, therefore $(b',s') \re{\eps} (b,s)$ in $U$.
    \end{enumerate}

    We conclude that either $(b,s)\re{\eps}(b',s')$ or $(b',s') \re{\eps} (b,s)$, as required.
\end{claimproof}

%% file: finite-to-infinite.tex
We now adapt Proposition~\ref{prop:existenceEpsComplete} to the finitary setting, for "$\omega$-regular" objectives.

\begin{proposition}\label{prop:existenceEpsComplete-finitary}
    Let $W$ be an "$\omega$-regular" objective recognised by a finite automaton $\A$ of index $d$.
    Assume that $W$ has ("chromatic@@mem") memory $\leq k$ on games of size $\leq 2^\kappa$, where 
    $$\kappa=k^{|Q|^2(|\Sigma| +1)} (3^{{k^2|Q|^2 (|\Sigma|+1)}^{d/2+1}}2|Q|^2)^{k2^{|Q|^{2}}}.$$
    Then there exists a "$k$-wise $\eps$-complete" automaton of size $\leq k |Q|$ recognising $W$.
\end{proposition}

Note that the size of $\kappa$ is doubly exponential in the size of $\A$, and therefore the bound on the size of the game graphs is triply exponential.
This section is devoted to proving the proposition.
The general idea of the proof is similar to the previous section (see Section~\ref{sec:overview} for a detailed overview), except that we now manipulate finite objects.
In particular, the notion of being (everywhere) cofinal is now irrelevant, and we will need more careful extraction mechanisms.

We will rely on the following standard lemma, which allows to bound the size of each block in Figure~\ref{fig:backtrack}. 
\AP We say that an accepting run in a given automaton is ""$\ell$-simple"" if its accepting priority $x$ occurs in every infix of size $\leq \ell$.

\begin{lemma}\label{lem:simple-runs}
Let $\A$ and $\B$ be automata of size $\leq n$.
Assume that every "$(2n^2)$-simple" accepting run in $\B$ is labelled by a word from $L(\A)$.
Then $L(\B) \subseteq L(\A)$.
\end{lemma}
\begin{proof}
    Let $\A\times \B$ be the product automaton of $\A$ and $\B$. That is, $\A\times \B$ is the graph with states $Q^\A\times Q^\B$ and transitions $(q,p)\re{a:(y,z)}(q',p')$ if $q\re{a:y}q'$ in $\A$ and $p\re{a:z}p'$ in $\B$.
    We say that a run $(q_1,p_1)\re{a_1:(y_1,z_1)}(q_2,p_2)\re{a_2:(y_2,z_2)}\dots$ in $\A\times \B$ is $\A$-accepting if $y_1y_2\dots \in \Parity$, and $\B$-accepting if $z_1z_2\dots \in \Parity$.

    Assume by contradiction that there is a run $\rho$ as above that is $\B$-accepting but $\A$-rejecting.
    Let $y_{\min} = \liminf y_i$ (which is odd) and $z_{\min} = \liminf z_i$ (which is even).
    \AP In the following, the term ""subrun"" means a subword of a run  (that is, a path obtained by removing edges of the given run) that is also a run.
    We will build a "subrun" of $\rho$ having as accepting priorities $y_{\min}$ and $z_{\min}$, in each component, but such that the projection into $\B$ is a "$(2n^2)$-simple" run, leading to a contradiction.
    
    Note that for every non-empty finite run $\upsilon = (q,p)\rp{u}(q',p')$ in $\A\times \B$, there exists a "subrun" $\upsilon' = (q,p)\rp{u'}(q',p')$ with same starting and ending point and of size $\leq n^2$. This is obtained just by recursively removing internal cycles (i.e. proper infixes starting and ending in the same state).
    Also, if $\upsilon = (q,p)\rp{u}(q',p')$ is a finite run containing one occurrence of priority $y_{\min}$ in the first component, there exists a "subrun" $\tilde{\upsilon} = (q,p)\rp{u'}(q',p')$ also containing some occurrence of $y_{\min}$ and of size $\leq 2n^2$.
    Indeed, it suffices to pick one occurrence of $y_{\min}$ in $\upsilon$ , and apply the previous remark to the prefix and the suffix of the run split in the occurrence of $y_{\min}$. That is, if $\upsilon = \upsilon_1 e \upsilon_2$, with $y_{\min}$ appearing in the first component of $e$, then $\tilde{\upsilon} = \upsilon_1' e \upsilon_2'$. 
    As we can assume that $e$ does not appear in $\upsilon_1'$ nor in $\upsilon_2'$ (by the previous argument), we have $|\tilde{\upsilon}| \leq 2(n^2-1)+1$.

    Consider the decomposition of the run $\rho = \upsilon_1 e_1  \upsilon_2 e_2  \dots$, where each transition $e_i$ carries priority $z_{\min}$ in the second component.
    The run $\rho' = \widetilde{\upsilon_1} e_1  \widetilde{\upsilon_2} e_2  \dots$ is as desired.
\end{proof}

As in the infinitary proof, we let $\Sgraph$ denote $\Skappad$ (which is finite), and apply Theorem~\ref{thm:structuration} to the product $\A \casc S$.
This gives a "$k$-blowup"  $G$ of $\A \casc S$ which is "$k$-wise $\eps$-complete" and  "satisfies" $W$.
As previously, we assume without loss of generality that $\A$ contains all transitions of the form $q \re{\eps:d-1} q$, and therefore for $s>s'$ in $\Sgraph$, since $s \re {d-1} s'$ in $\Sgraph$ we have edges $(q,m,s) \re{\eps} (q,m,s')$ in $G$.

We again let $\overline G$ be obtained by closing $G$ by transitivity, i.e. adding transitions $(q,m,s) \re c (q',m',s')$, for $c \in \Sigmaeps$, whenever $(q,m,s) \re{\eps^* c \eps^*} (q',m,s')$ in $G$.
Note that $\overline G$ also  "satisfies" $W$.
Without loss of generality (up to removing some $\eps$-edges while preserving being "$k$-wise $\eps$-complete"), we assume that for each $m$, $\re \eps$ defines a linear order over $Q \times \{m\} \times S$ in $\overline G$.

\subsubsection{Domination and definition of the automaton}\label{sec:domination-and-overview-finitary}

As in the previous proof, in order to define $\B$, we rely on a notion of domination; however we now need a more precise definition.
Consider a "signature"\footnote{Throughout the proof, we drop the word ``signature'' and only talk of ``trees'' since these are the only kind of trees we consider here.} tree $T \subseteq \kappa^{\d_\odd}$ with "branching" $b$; the following definitions are with respect to $T$, but we omit $T$ in the notations for clarity.
Let $s_{<x}$ be a node at level $x<d$ in $T$ and let $s_{<x}s_{x+1}^{(0)},\dots,s_{<x}s_{x+1}^{(b-1)}$ denote its children.
Recall that $0_{>x}$ denotes the sequence with $0$s indexed at odd positions $>x$.
Then we say that $(q,m)$ strongly $c$-dominates $(q',m')$ at $(s_{<x},i)$ if
\[
    (q,m,s_{<x} s_{x+1}^{(0)} 0_{>x+2}) \re c (q',m',s_{<x} s_{x+1}^{(i)}0_{>x+2}) \tin \overline G
\]
and that $(q,m)$ weakly $c$-dominates $(q',m')$ at $(s_{<x},i)$ if
\[
    (q,m,s_{<x} s_{x+1}^{(i+1)} 0_{>x+2}) \re c (q',m',s_{<x} s_{x+1}^{(i)}0_{>x+2}) \tin \overline G.
\]
We say that $(q,m)$ strongly (resp. weakly) $c$-dominates $(q',m')$ at $s_{<x}$ if this is the case for all $i \in b$, and that $(q,m)$ strongly (resp. weakly) $c$-dominates $(q',m')$ at level $x$ if this is the case for all nodes $s_{<x}$ at level $x$.

Given a "tree@@sign" $T \subseteq \kappa^{\d_\odd}$, we define an automaton $\B_T$ in the same way as in the previous section, with the adapted definitions of domination:
We put $V(\B_T) = Q \times k$, and for each even $x \in d$ and $c \in \Sigmaeps$, we define transitions by
\[
    \begin{array}{lcrl}
    (q,m) &\re{c:x}& (q',m') & \tif (q,m) \text{ strongly $c$-dominates $(q',m')$ at level $x$ in $T$}, \\
    (q,m) &\re{c:x+1}& (q',m') & \tif (q,m) \text{ weakly $c$-dominates $(q',m')$ at level $x$ in $T$}.
    \end{array}
\]

To obtain a correct automaton $\B_T$, the challenge lies in extracting a "tree@@sign" $T$ whose "branching" matches the size $2|Q|^2$ of blocks given by Lemma~\ref{lem:simple-runs}, and such that the following holds.
\begin{enumerate}
    \item\label{item:blowup} For every $(q,m) \in V(\B_T)$, every transition $q \re {c:x} q'$ in $\A$ (resp. $\re {c:x+1}$) and every node $s_{<x}$ at level $x$, there is $m'$ such that $(q,m)$ strongly (resp. weakly) $c$-dominates $(q',m')$ at $s_{<x}$.
    This will ensure that $\B_T$ is indeed a "$k$-blowup"  of $\A$.
    \item\label{item:eps-complete} For every "memory state" $m$, every node $s_{<x}$ at level $x$ and every pair of states $(q,q')$, either $(q,m)$ strongly $\eps$-dominates $(q',m)$ or $(q',m)$ weakly $\eps$-dominates $(q,m)$.
    This will ensure that $\B_T$ is $\eps$-complete.
    \item\label{item:uniformity} The fact that $(q,m)$ strongly (or weakly) $c$-dominates $(q',m')$ at $s_{<x}$ depends only on the level $x$, but not on the choice of the node $s_{<x}$.
    That is, in the sought tree $T$, $(q,m)$ strongly (or weakly) $c$-dominates $(q',m')$ at level $x$ if this is the case \emph{for some} node in level $x$.
\end{enumerate}

A first extraction allows to guarantee item~\ref{item:blowup}; this is shown in Section~\ref{sec:blowup-extraction} below.
However, choosing $T$ such that item~\ref{item:eps-complete} holds requires some combinatorics which are presented in Section~\ref{sec:regular-extraction}.
We then adapt the previous cofinal extraction lemma (Lemma~\ref{lem:everywhere_cofinal}) to the finitary setting, and apply it to obtain the third item in Section~\ref{sec:uniformity}.
Proposition~\ref{prop:existenceEpsComplete-finitary} is finally established in Section~\ref{sec:correctness-of-B-finitary}.

\subsubsection{Guaranteeing that $\B$ is a $k$-blowup}\label{sec:blowup-extraction}

We prove the first item from the explanation above.

\begin{lemma}\label{lem:guaranteeing-k-blowup-finitary}
    There exists a "tree@@sign" $T \subseteq \kappa^{\d_\odd}$ with "branching" $$b_1=\kappa/k^{|Q|^2(|\Sigma| +1)}=(3^{{k^2|Q|^2 (|\Sigma|+1)}^{d/2+1}}2|Q|^2)^{k2^{|Q|^{2}}}$$ such that for every $(q,m) \in Q \times k$ and every transition $q \re{c:x} q'$ in $\A$ (resp. $\re{c:x+1}$) and every node $s_{<x}$ at level $x$ in $T$, there is $m'$ such that $(q,m)$ strongly (resp. weakly) $c$-dominates $(q',m')$ at $s_{<x}$.
\end{lemma}
    
\begin{proof}
    First, observe that if $j \geq i$ and $i' \geq j'$, and 
    \[
        (q,m,\underbrace{s_{<x}s_{x+1}^{(i)} 0_{>x+2}}_{s^i}) \re{c} (q',m',\underbrace{s_{<x}s_{x+1}^{(i')}0_{>x+2}}_{s^{i'}}) \tin \overline G
    \] 
    then also 
    \[
        (q,m,\underbrace{s_{<x}s_{x+1}^{(j)} 0_{>x+2}}_{s^j}) \re{c} (q',m',\underbrace{s_{<x}s_{x+1}^{(j')}0_{>x+2}}_{s^{j'}}),
    \]
    because $s^{j} > s^{i'}$ and $s^{i'} > s^{j'}$. 

    We prove the following claim, and then explain how the lemma follows.

    \begin{claim}
        Let $T$ be a "tree@@sign" of "branching" $b$, let $(q,m) \in Q \times k$, let $q \re{c:y} q'$ be a transition in $\A$ and let $s_{<x}$ be a node at level $x$ in $T$, where either $y=x$ or $y=x+1$.
        \begin{itemize}
            \item If $y$ is even then there is $m' \in k$ such that $(q,m)$ strongly $c$-dominates $(q',m')$ at $s_{<x}$ (with respect to $T$).
            \item If $y$ is odd, then there is $m'$ and $b / k$ children of $s_{<x}$ such that in any subtree of $T$ where the children of $s_{<x}$ are among these ones, it holds that $(q,m)$ weakly $c$-dominates $(q',m')$ at $s_{<x}$.
        \end{itemize}
    \end{claim}
    \begin{claimproof}
    Let $s_{<x} s_{x}^{(0)},\dots, s_{<x} s_x^{(b-1)}$ denote the $b$ children of $s_{<x}$.
    There are two cases according to the parity of $y$.
    \begin{itemize}
    \item If $y=x$ is even.
    First, we prove that there is $m'$ such that $(q,m)$ strongly $c$-dominates $(q',m')$ at $(s_{<x},b-1)$, which means 
    \[
        (q,m,\underbrace{s_{<x} s_{x+1}^{(0)} 0_{>x+2}}_{s^{0}}) \re c (q',m',\underbrace{s_{<x} s_{x+1}^{(b-1)}0_{>x+2}}_{s^{b-1}}) \tin \overline G.
    \]
    Indeed, by definition of $\Sgraph$ we have $s^{0} \re{x} s^{b-1}$ in $\Sgraph$ therefore $(q,s^{0}) \re{c} (q',s^{b-1}) \tin \A \casc S$ and thus we get the wanted result since $G$ is a "$k$-blowup"  of $\A \casc S$.

    Now by the observation above, we also get that $(q,m)$ strongly $c$-dominates $(q',m')$ at $(s_{<x},i')$ for any $i' \leq b-1$,  and thus $(q,m)$ strongly $c$-dominates $(q',m')$ at $s_{<x}$.
    \item If $y=x+1$ is odd.
    As above, we first prove that for every $i$, there is $m'_i$ such that $(q,m)$ weakly $c$-dominates $(q',m'_i)$ at $(s_{<x},i)$, which means
    \[
        (q,m,\underbrace{s_{<x} s_{x+1}^{(i+1)} 0_{>x+2}}_{s^{i+1}}) \re c (q',m'_i,\underbrace{s_{<x} s_{x+1}^{(i)}0_{>x+2}}_{s^{i}}) \tin \overline G.
    \]
    Indeed, by definition of $\Sgraph$ we have $s^{i+1} \re{x+1} s^{i}$ in $\Sgraph$ therefore $(q,s^{i+1}) \re{c} (q',s^{i}) \tin \A \casc S$ and thus we get the wanted result since $G$ is a "$k$-blowup"  of $\A \casc S$.

    Now, by the pigeonhole principle, take $m'$ such that $m'=m'_i$ for at least $b / k$ values $i_0<i_1<\dots<i_{b/k-1}$ of $i$.
    By the observation above, we get that for every $j' < j< b/k$,
    \[
        (q,m,s^{i_{j}}) \re c (q',m',s^{i_{j'}`'}) \tin \overline G.
    \]
    Therefore in any "tree@@sign" satisfying that the children of $s_{<x}$ are a subset of the $s_{x}^{i_j}$'s, we get that $(q,m)$ weakly $c$-dominates $(q',m')$ at $s_{<x}$.\qedhere
    \end{itemize}
\end{claimproof}
    Therefore, given a node $s_{<x}$ at level $x$, and starting with some "tree@@sign" in which $s_{<x}$ has "branching" $\kappa$, we apply the claim to all transitions $q \re {c:x+1} q'$ in $\A$ successively, each time removing all nodes which are not (descendants of) the children given by the lemma.
    Since there are $\leq |Q|^2 |\Sigma \cup \{\eps\}|$ such transitions, the node $s_{<x}$ has "branching" $\geq \kappa / k^{|Q|^2 (|\Sigma| +1)}$ in the resulting tree.
    If this inequality is strict, we may remove more children arbitrarily so that it becomes an equality.
    Note that the branchings of other nodes that remain in the "tree@@sign" are not affected.

    Applying this first to $s_{<0}$ being the root, and then in a top-down fashion, gives the wanted result.
\end{proof}

Note that the conclusion of the lemma still holds when taking a subset of $T$; this will allow us to perform further extractions to guarantee more properties.

\subsubsection{Guaranteeing that $\B$ is $\eps$-complete: Regular words and embeddings}\label{sec:regular-extraction}

We now introduce some more notions which will allow us to ensure item~\ref{item:eps-complete} above.

\subparagraph*{Regular words.}
A word $w \in C^*$ is called an $n$-word if every letter from $C$ occurs exactly $n$ times in $w$ (later in the proof, we will take $C$ to be $Q$).
Recall that we denote $n=\{0,1,\dots,n-1\}$. 
Given an $n$-word $w$ and two letters $a,b$, we say that
\begin{itemize}
\item $a$ is strongly after $b$ (or $b$ is strongly before $a$), if all occurrences of $a$ are after all occurrences of $b$ in $w$;
\item $a$ is weakly after $b$ (or $b$ is weakly before $a$), if for every $i<j$, the $j$-th occurrence of $a$ is after the $i$-th occurrence of $b$; and
\item $a$ and $b$ are interleaved, if $a$ is both weakly after and weakly before $b$.
\end{itemize}
For instance, in
\[
    \us a 0 \us a 1 \us a 2 \us c 0 \us c 1 \us a 3 \us c 2 \us a 4 \us b 0 \us b 1 \us b 2 \us b 3 \us c 3 \us c 4 \us b 4,
\]
we have that $a$ is strongly before $b$, and $c$ is weakly before $b$, and in
\[
    \us a 0 \us b 0 \us b 1 \us a 1 \us a 2 \us b 2,
\]
letters $a$ and $b$ are interleaved.

An $n$-word is called regular if for every pair of letters $(a,b)$, either $a$ is strongly before $b$ or $b$ is weakly after $a$.
This amounts to saying that for every unordered pair of letters $\{a,b\}$, either one is strongly before the other, or they are interleaved (applying the previous property to both $(a,b)$ and $(b,a)$).
This also amounts to saying that the alphabet can be partitioned into $p$ parts, and the word can be broken into $p$ corresponding subsequent blocs, where each block is an interleaving of the letters of the corresponding part.
For instance,
\[
    abbabacccdedeedf\! f\! f
\]
is a regular $3$-word, where the partition is $\{a,b\} \sqcup \{c\} \sqcup \{d,e\} \sqcup \{f\}$.

\subparagraph*{Word-embeddings and extraction lemma.}
A word-embedding from $w$ to $w'$ is a set $R$ of positions in $w'$, such that restricting $w'$ to positions in $R$ gives the word $w$.
We display in bold which letters of $w'$ belong to $R$, for instance $\ul a ba \ul b$ represents the embedding given by $R=\{0,3\}$ (recall that we number positions from $0$) and therefore $w=ab$.
In this case, we say that $w$ is a subword of $w'$.

Given a subset $P$ of $n'$ 
and an $n'$-word $w'$, we may define the embedding given by keeping the $i$-th occurrence of each letter, for each $i \in P$.
For instance, if $P=\{1,4,5\}$, we get the following embedding
\[
    \us{a}{0} \us{b}{0} \us {\ul a} 1 \us{\ul b} 1  \us a 2 \us a 3 \us b 2 \us{b} 3  \us {\ul a} 4 \us{\ul a} 5 \us a 6 \us {\ul b} 4  \us {\ul b} 5 \us b 6.
\]
In this case, we say that the obtained word $w$ is the $P$-subword of $w'$, and we say that $w$ is a synchronised subword of $w'$ if there is such a $P$.
Note that in this case $w$ is an $n$-word for $n=|P|$.
The following easy lemma will often be used below.

\begin{lemma}\label{lem:synchronised-subwords-are-regular}
    A synchronised subword of a regular word is regular.
\end{lemma}

We are now ready to state our main extraction lemma.

\begin{lemma}\label{lem:extraction-regular-words}
Let $n,k$ be positive integers, fix a finite alphabet $C$, let $w'_0,\dots,w'_{k-1}$ be $n'$-words with $n' \geq n^{k2^{|C|^{2}}}$. 
There is a set $P \subseteq n'$ of size $n$ such that the $P$-subwords $w_0,\dots,w_{k-1}$ of $w'_0,\dots,w'_{k-1}$ are regular.

\end{lemma}
\begin{proof}
    In this proof, we write $a^{(i)} <_w b^{(j)}$ to say that the $i$-th $a$ occurs before the $j$-th $b$ in the word $w$.
    \begin{itemize}
        \item \textbf{One word, two letters.} First assume that there is a single word $w'$ and $|C|=2$. In this case, we will show that it suffices to assume that $n' \geq n^2$.
        \begin{itemize}
            \item If there is $i$ such that for some letter $a\in C$, it holds that $a^{(i+n-1)}<_{w'}b^{(i)}$, where $b$ is the other letter.
            Then for $P=\{i,i+1,\dots,i+n-1\}$, we get that $a$ is strongly before $b$ in the $P$-subword $w$ of $w'$, and thus $w$ is regular (and has size $n$).
            \item Otherwise, we let $P=\{0,n,2n,\dots,(n-1)n\}$, and claim that both letters are interleaved in the $P$-subword $w$ of $w'$.
            Indeed, we should prove that for any letter $a$, it holds that $a^{((i+1)n)} >_{w} b^{(in)}$, and this is because we are not in the first case.
        \end{itemize}
        \item \textbf{One word, many letters.} We now assume that there is still just one word $w$, but $C$ is any finite set and $n'=n^{2^{|C|^2}}$.
        Say that an unordered pair of different letters is good in a given word if either one is strongly before the other, or they are interleaved.
        Just as for Lemma~\ref{lem:synchronised-subwords-are-regular}, observe that if a pair of letters is good in a given word, then it is also good in its synchronised subwords.

        We proceed as follows.
        Choose any (unordered) pair of different letters $a,b$, and apply the previous case to get $P_{0}$ of size $(n')^{1/2}$ such that in the $P_0$-subword $w_{0}$, the pair $a,b$ is good.
        Then take another pair $a',b'$ and apply the same reasoning to $w_0$: we get $P_1 \subseteq P_{0}$ such that both $a',b'$ and $a,b$ (by the above observation) are good in the $P_1$-subword of $w$.
        Then repeat this process for each of the $|C|(|C|-1)/2\leq |C|^2$ unordered pairs of letters.
        In each step, we get that $|P_i| \geq |P_{i-1}|^{1/2}$, by the previous case.
        
        \item \textbf{Many words, many letters.}
        We are now ready to prove the lemma.
        First we apply the previous case to $w'_0$, which yields $P_0$ such that the $P_0$ subword of $w'_0$ is regular.
        Then we apply the previous case to the $P_0$-subword of $w'_1$: this gives $P_1 \subseteq P_0$ such that the $P_1$-subword of $w'_1$ is regular.
        By the observation above, since $P_1 \subseteq P_0$, the $P_1$-subword of $w'_0$ is also regular.
        We proceed in this fashion for $i=2,3,\dots,k-1$, which proves the lemma. \qedhere
    \end{itemize}
\end{proof}

\subparagraph*{Guaranteeing $\eps$-completability.}
Recall that the definitions of strong and weak domination (see beginning of Section~\ref{sec:domination-and-overview-finitary}) are implicitely parameterised by the choice of a "tree@@sign" $T$.
We are now ready to prove the requirement over $T$, which will ensure "$\eps$-completability" of the obtained automaton $\B_T$.

\begin{lemma}\label{lem:ensuring-eps-completability}
There exists a "tree@@sign" $T \subseteq \kappa^{\d_\odd}$ with "branching" $$b_2=(b_1)^{1/(k2^{|Q|^{2}})}=3^{{k^2|Q|^2 (|\Sigma|+1)}^{d/2+1}}2|Q|^2$$ such that the conclusion of Lemma~\ref{lem:guaranteeing-k-blowup-finitary} holds and moreover for every "memory state" $m \in k$, every node $s_{<x}$ at level $x$ and every pair of states $(q,q')$, either $(q,m)$ strongly $\eps$-dominates $(q',m)$ or $(q',m)$ weakly $\eps$-dominates $(q,m)$.
\end{lemma}

\begin{proof}
First let $T'$ be a "tree@@sign" of "branching" $b_1$ obtained from Lemma~\ref{lem:guaranteeing-k-blowup-finitary}.
Fix a node $s_{<x}$ at level $x$ in $T'$.
Then for each "memory state" $m \in k$, consider the word $w'_m\in Q^*$ indexed by vertices of $\overline G$ of the form $(q,m,s_{<x} s_{x+1} 0_{>x+2})$ such that $ s_{<x} s_{x+1}$ is a child of $s_{<x}$ in $T'$.
These vertices, which correspond to the positions in $w'_m$, are linearly ordered by $\re \eps$ over $\overline G$, and the letters correspond to the projections on the first coordinate.
Note that these are $b_1$-words: each state has exactly $b_1$ occurrences. 

We apply Lemma~\ref{lem:extraction-regular-words} to the words $w'_0,\dots,w'_{k-1}$, which yields a set $P \subseteq b_1$ of size $b_2$ such that the corresponding $P$-subwords are regular.
This means that if we set the $b_2$ children of $s_{<x}$ to be the $s_{x+1}^{(p)}$ for $p \in P$, then for each pair of states $(q,q')$ (i.e. letters), either $q$ is strongly before $q'$, in which case $q$ $\eps$-strongly dominates $q'$ at $s_{<x}$, or $q'$ is weakly after $q$, in which case $q'$ $\eps$-weakly dominates $q$ at $s_{<x}$.

Let $P^{s_{<x}}$ be the set $P$ as above obtained with respect to the node $s_{<x}$ of $T'$.
We produce a tree satisfying the conclusion of the lemma by succesively restricting the subtree at $s_{<x}$ to the children in $P^{s_{<x}}$ 
\end{proof}

Note also that thanks to Lemma~\ref{lem:synchronised-subwords-are-regular}, the conclusion of Lemma~\ref{lem:ensuring-eps-completability} still holds for trees that are subsets of $T$, so we may still extract trees with more properties.

\subsubsection{Making $T$ homogeneous}\label{sec:uniformity}

We now show how to ensure item~\ref{item:uniformity} from the overview (Section~\ref{sec:domination-and-overview-finitary}).

\subparagraph*{Extracting highly "branching" trees.}
We need another extraction lemma, which is a straightforward adaptation of Lemma~\ref{lem:everywhere_cofinal} to the finitary setting.
Recall that an "inner labelling" is a map from nodes of $T$ to some set $L$, and that it is "constant per level" if for every $x\in\dodd$, $\lambda$ is constant over nodes of level $x$ in $T$.

\begin{lemma}\label{lem:highly-branching-extraction}
    Let $b$ be a positive integer and consider a "tree@@sign" $T'$ with "branching" $b' = b|L|^{d/2+1}$ and an "inner labelling" $\lambda$ of $T'$ by $L$.
    There is a "tree@@sign" $T \subseteq T'$ with "branching" $b$ such that $\lambda|_{T}$ is "constant per level".
\end{lemma}

\begin{proof}
    The proof is exactly the same as for Lemma~\ref{lem:everywhere_cofinal}, by induction on $d$.
    We still reproduce it for clarity.

    For $d=0$ there is nothing to prove since $\dodd$ is empty; let $d\geq 1$ and assume the result known for $d-2$.
    For each node $s_{<2}$ at level $2$, apply the induction hypothesis on the subtree of $\kappa^{\dodd}$ rooted at $s_{<2}$,
    which gives a tree $T_{s_{<2}}$ with "branching" $b$ such that $\lambda$ is "constant per level" over $T_{s_{<2}}$.
    Let $\ell^{s_{<2}}_2,\dots,\ell^{s_{<2}}_{d}$ denote the constant values of $\lambda$ on the corresponding levels of $T_{s_{<2}}$, and define a new auxiliary labelling of the nodes $s$ at level $2$ of $\Sgraph$ by the tuple $\ell = (\lambda(s),\ell^{s_{<2}}_2,\dots,\ell^{s_{<2}}_{d})$.

    Now since there are at most $|L|^{d/2+1}$ new auxiliary labels, and there are $b'$ nodes at level $2$, there is a new labelling such that $\geq b$ nodes have this new label.
    We conclude by taking $T$ to be the union of $\{s\} \times T'_s$, where $s$ ranges over nodes at level $2$ with the new labelling $\ell$, and $T'_s$ are the corresponding everywhere cofinal trees.
\end{proof}


We are now ready to prove the wanted statement.

\begin{lemma}\label{lem:guaranteeing-uniformity}
There exists a "tree@@sign" $T \subseteq \kappa^{\d_\odd}$ with "branching" $$b_2/3^{{k^2|Q|^2 (|\Sigma|+1)}^{d/2+1}}=2|Q|^2$$ such that the conclusion of Lemmas~\ref{lem:guaranteeing-k-blowup-finitary} and~\ref{lem:ensuring-eps-completability} hold for $T$ and moreover, the fact that $(q,m)$ strongly (or weakly) $c$-dominates $(q',m')$ at $s_{<x}$ depends only on the level $x$ but not on the choice of the node $s_{<x}$.
\end{lemma}


\begin{proof}
    Let $T'$ be a "tree@@sign" of "branching" $b_2$ obtained from Lemma~\ref{lem:ensuring-eps-completability}.
    Consider the "inner labelling" which assigns to each node $s_{<x}$ of level $x$ in $T'$, the information, for each (ordered) pair of nodes $(q,m)$, $(q',m')$ and each $c \in \Sigma \cup \{\eps\}$, whether $(q,m)$ strongly, or weakly, or not at all, $c$-dominates $(q',m')$.
    This defines an "inner labelling" of $T'$ by $L$ where $|L|=3^{k^2|Q|^{2}(|\Sigma|+1)}$.
    Therefore applying Lemma~\ref{lem:highly-branching-extraction} concludes.    
\end{proof}

We are now ready to define $\B$.

\subsubsection{Correctness of $\B$}\label{sec:correctness-of-B-finitary}
We are finally ready to prove Proposition~\ref{prop:existenceEpsComplete-finitary}.
Fix a "tree@@sign" $T$ with "branching" $2|Q|^2$ as given by Lemma~\ref{lem:guaranteeing-uniformity}, and let $\B = \B_T$ (see definition in Section~\ref{sec:domination-and-overview-finitary}).

The facts that $\B$ is a "$k$-blowup"  of $\A$ and that it is "$\eps$-complete" are rephrasings of Lemmas~\ref{lem:guaranteeing-k-blowup-finitary} and~\ref{lem:ensuring-eps-completability} respectively.
There remains to prove that $\B$ recognises $W$ (and that $\B$ is indeed an automaton in the sense that it does not have an accepting run over $\Sigma^*\eps^\omega$, which, as in Section~\ref{sec:existence_automata}, will be proved at the same time).
This is done by adapting Lemmas~\ref{lem:technical1},~\ref{lem:technical2} and~\ref{lem:language-containement} from Section~\ref{sec:existence_automata} to the current construction.

\begin{lemma}\label{lem:odd-transitions-finitary}
    Let $x \in d$ be even, assume that $(q,m) \re{c:y} (q',m')$ in $\B$ for some $y>x$ (which can be even or odd), let $i \in 2|Q|^2$ and let $s_{<x}$ be a node at level $x$ in $T$.
    Then
    \[
        \underbrace{(q,m,s_{<x} s_{x+1}^{(i+1)}0_{>x+2})}_{u} \re c \underbrace{(q',m',s_{<x} s_{x+1}^{(i)}0_{>x+2})}_{u'} \tin \overline G.
    \]
\end{lemma}

\begin{proof}
    Let $x'$ be even such that $y \in \{x',x'+1\}$; note that $x \leq x'$.
    By definition of $\B$ (independently of the parity of $y$), $(q,m)$ weakly $c$-dominates $(q',m')$ at level $x'$, so
    \[
        \underbrace{(q,m,s_{<x'} s_{x'+1}^{(i+1)} 0_{>x'+2})}_{v} \re c \underbrace{(q',m',s_{<x'} s_{x'+1}^{(i)} 0_{>x'+2})}_{v'}, \tag*{(1)}\label{eq:1}
    \]
    for every $i$ and $s_{<x'} = s_{<x} s_{x+1}^{(i)} 0_{(x+2,x')}$, where $0_{(x+2,x')}$ denotes the all-zero tuple indexed by odd numbers between $x+2$ and $x'$ (if any).
    Since moreover we have 
    \[
        s_{<x} s_{x+1}^{(i+1)} 0_{>x+2} > s_{<x'} s_{x'+1}^{(i+1)} 0_{>x'+2} \tag*{(2)}\label{eq:2}
    \]
    and 
    \[
        s_{<x'} s_{x'+1}^{(i)} 0_{>x'+2} > s_{<x'} 0_{>x+2} \tag*{(3)}\label{eq:3}
    \]
    which respectively give edges
    \[
        \underbrace{(q,m,s_{<x} s_{x+1}^{(i+1)} 0_{>x+2})}_{u} \re \eps \underbrace{(q,m,s_{<x'} s_{x'+1}^{(i+1)} 0_{>x'+2})}_{v} \tin \overline G
    \]
    and
    \[
        \underbrace{(q',m',s_{<x'} s_{x'+1}^{(i)} 0_{>x'+2})}_{v'} \re \eps \underbrace{(q',m',s_{<x'} 0_{>x+2})}_{u'} \tin \overline G,
    \]
    we get the wanted edge $u \re c u'$ in $\overline G$ by transitivity. 
\end{proof}

\begin{lemma}\label{lem:even-transitions-finitary}
    Assume that $(q,m) \re{c:x} (q',m')$ in $\B$ for some even $x$, and let $s_{<x}$ be a node at level $x$ in $T$.
    Then for every $i \in 2|Q|^2$, it holds that
    \[
        (q,m,s_{<x}s_{x+1}^{(0)}0_{>x+2}) \re c (q',m',s_{<x} s_{x+1}^{(i)}0_{>x+2}).
    \]
\end{lemma}

\begin{proof}
    This is exactly the definition of $(q,m) \re{c:x} (q',m') \tin \B$ for even $x$.
\end{proof}

We are now ready to prove the final result.

\begin{lemma}
    The language of $\B$ is contained in $W$. Moreover, there is no accepting run labelled by a word in $\Sigma^* \eps^\omega$.
\end{lemma}

The proof is just like for Lemma~\ref{lem:language-containement} in Section~\ref{sec:existence_automata}.

\begin{proof}
    Thanks to Lemma~\ref{lem:simple-runs}, it suffices to prove that accepting "$2|Q|^2$-simple" runs are labelled by words whose projection on $\Sigma$ is infinite and belongs to $W$.
    Take such an accepting run 
    \[
        (q_0,m_0) \re{c_0:y_0} (q_1,m_1) \re{c_1:y_1} \dots
    \]
    in $\B$.
    Let $x=\liminf_i y_i$ (it is even since the run is accepting), and let $i_0$ be such that $y_i \geq x$ for $i \geq i_0$.

    As previously our goal is to endow each $(q_i,m_i)$ with some $s_i \in T$ such that for all $i$, $(q_i,m_i,s_i) \re{c_i} (q_{i+1},m_{i+1},s_{i+1}) \tin \overline G$.
    This implies the result since $\overline G$  "satisfies" $W$, and since it does not have paths labelled by words in $\Sigma^* \eps^\omega$ by well-foundedness.

    To define the $s_i$'s, we proceed as follows.
    Let $s_{< x}$ denote an arbitrarily node at level $x$ in $T$, and let $s_{<x} s_{x+1}^0,\dots, s_{<x} s_{x+1}^{2|Q|^2-1}$ denote its $2|Q|^2$ children in $T$.
    Then for each $i$, we let $j_i$ be the least integer $\geq 0$ such that $y_{i+j_i} = x$, and observe that $j_i<2|Q|^2$ since the run is "$2|Q|^2$-simple".
    Then for each $i$, we set $s_i=s_{<x} s_{x+1}^{j_i} 0_{>x+2}$.

    For $i$'s such that $j_i>0$, we get that indeed $(q_i,m_i,s_i) \re {c_i} (q_{i+1},m_{i+1},s_{i+1})$ thanks to Lemma~\ref{lem:odd-transitions-finitary}.
    For $i$'s such that $j_i=0$, the same is true thanks to Lemma~\ref{lem:even-transitions-finitary}.
    This concludes the proof.
\end{proof}

As before, the proof adapts directly to the "chromatic@@mem" case: if $G$ is "chromatic@@graph" with update $\chi$ then so is $\B$.

%% file: union.tex

In this section, we establish a strong form of the "generalised Kopczyński conjecture" for $\BCSigma$ objectives.
Recall that an objective $W \subseteq \Sigma^\omega$ is "prefix-increasing"\footnote{In other papers~\cite{Ohlmann21PhD,Ohlmann23,CO25LMCS} this notion is called prefix-decreasing, as Eve is seen as a ``minimiser'' player who aims to minimise some quantity.}
 if for all $a \in \Sigma$ and $w \in \Sigma^\omega$, it holds that if $w \in W$ then $aw \in W$.
In words, one remains in $W$ when adding a finite prefix to a word of $w$.
Examples of "prefix-increasing" objectives include prefix-independent and closed objectives.

\thmUnion*


\begin{remark}
The assumption that one of the two objectives is "prefix-increasing" is indeed required: for instance if $W_1=aa(a+b)^\omega$ and $W_{2}=bb(a+b)^\omega$, which are positional but not "prefix-increasing", the union $(aa+bb)(a+b)^\omega$ is not positional (it has "memory" $2$).

\end{remark}
\begin{remark}
    The bound $k_1k_2$ in Theorem~\ref{thm:union} is tight: For every $k_1, k_2$, there are objectives $W_1$, $W_2$ with memories $k_1, k_2$ respectively, such that $W_1 \cup W_2$ has "memory" exactly $k_1k_2$. 
    One such example is as follows: let $\SS = \{a_1,\dots, a_{k_1}, b_1,\dots, b_{k_2}\}$ and 
    $W_1 = \{ w \mid w \text{ contains at } \text{least } \text{two } \text{different } a_i \text{ infinitely often} \}$ and $W_2 = \{ w \mid w$ contains at least two  $b_i$ infinitely often$\}$. 
    We can see that $W_1$, $W_2$ and $W_1\cup W_2$ have memory, respectively, $k_1, k_2$, and  $k_1\cdot k_2$ by building the Zielonka tree of these objectives and applying~\cite[Thms.~6,~14]{DJW1997memory}.
\end{remark}

The rest of the section is devoted to the proof of Theorem~\ref{thm:union}.
We explicit a construction for the union of two "parity automata", inspired from the Zielonka tree of the union of two parity conditions, and show that it is "$(k_1 k_2)$-wise $\eps$-completable".

\subparagraph{Union of parity languages.} We give an explicit construction of a deterministic parity automaton $\mathcal{T}$ recognising the union of two parity languages, which may be of independent interest.  This corresponds to the automaton given by the Zielonka tree of the union (a reduced and structured version of the LAR).

We let $\done = \{0, 1,\dots,d_1\}$  and $\dtwo= \{1^*,\dots,d_2^*\}$.
Let $\doneodd$ and $\dtwoodd$ denote the restrictions of $\done$ and $\dtwo$ to odd elements.
By a slight abuse of notation, we sometimes treat elements in $\dtwo$ as natural numbers (e.g. when comparing them).\\

\textit{Alphabet.} 
The input alphabet is $\done \times \dtwo$, and we write letters as $(y,z)$. 
The "index" of $\T$  is at most $d_1 + d_2+2$, we use the letter $t$ for its output priorities.\\

\textit{States}. 
States are given by interleavings of the sequences $\langle 1,3,\dots,d_1\rangle$ and $\langle 1^*,3^*,\dots,d_2^*\rangle$  of the elements in $\doneodd$ and $\dtwoodd$ (assuming $d_1, d_2$ odd). Any state can be taken as initial. For instance, for $d_1=d_2=5$, an example of a state is:\\[-4mm]
\[  \hspace{15mm} \tau = \langle \tikzmark{tau1}1,\tikzmark{tau2}3,\tikzmark{tau3}1^*,\tikzmark{tau4}5, \tikzmark{tau5}3^*, \tikzmark{tau6}5^* \rangle .\]

\begin{tikzpicture}[remember picture, overlay] 
    \foreach \lbl in {1,...,6} {
        \node[below, xshift = 1.4mm, yshift = -1mm, scale=0.7, color = lipicsGray] at (pic cs:tau\lbl) {$\lbl$};
    }
\end{tikzpicture}


\vspace{-2.5mm}

We use $\tau$ to denote such a sequence, which we index from $1$ to $(d_1+d_2)/2$, and write $\tau[i]$ for its $i$-th element.
\AP For $y\in \done$, $y$ odd,  we let $\intro*\indtau{y} = i$ for the index such that $\tau[i]=y$. For $y\geq 2$ even, we let $\indtau{y}$ be the index $i$ such that $\tau[i]=y-1$. We let $\indtau{0} = 0$. We use the same notation for $z\in \dtwo$, $z\geq 1$. For example, in the state above, $\indtau{2^*}=3$.

The intuition is that a state stores a local order of importance between input priorities.\\

\textit{Transitions.}
Let $\tau$ be a state and $(y,z)$ an input letter. We define the transition $\tau \re{(y,z):t}\tau'$ as follows:

Let $i = \min \{ \indtau{y}, \indtau{z}\}$. 
In the following, we assume $i = \indtau{y}$ (the definition for $i = \indtau{z}$ is symmetric). 
We let $t = 2i$, if $y$ even, and  $t = 2i-1$ if $y$ is odd.
If $y$ is even, we let $\tau' = \tau$.
If $y$ is odd, let $i'$ be the smallest index $i<i'$ such that $\tau[i']\in \dtwo$, and let $\tau'$ be the sequence obtained by inserting $\tau[i']$ on the left of $\tau[i]$ (or $\tau' = \tau$ if no such index $i'$ exists). 
Formally,
\[ \tau'[j] = \tau[j] \tfor j<i \tand i'<j,  \quad \tau'[i] = \tau[i'] \quad \tand \quad \tau'[j] = \tau[j-1] \tfor i<j\leq i' .\]

For example, for the state above, we have:
$ \langle 1,3,1^*, 5, 3^*, 5^* \rangle  \re{ (3,2^*):3} \langle 1, 1^*, 3, 5, 3^*, 5^* \rangle. $\\

For $w=(y_1,z_1)(y_2,z_2)\dots\in (\done \times \dtwo)^\omega$, we let $\pi_1(w) = y_1y_2\dots$ and $\pi_2(w) = z_1z_2\dots$.

\begin{lemma}\label{lem:language-of-T}
    The automaton $\T$ recognises the language \[L = \{w \in (\done \times \dtwo)^\omega \mid \pi_1(w) \in \Parity_{d_1} \ \mathrm{or} \ \pi_2(w) \in \Parity_{d_2}\}.\]
\end{lemma}
\begin{proof}
    We show that $L\subseteq L(\T)$, the other inclusion is similar (and implied by Lemma~\ref{lem:overline-T} below).
    Let $(y_1,z_1)(y_2,z_2)\dots \in L$, and assume w.l.o.g. that  $y_1y_2\dots \in \Parity_{d_1}$.
    Let $\ymin = \liminf y_i$, which is even, and let $n_0$ be so that for all $n \geq n_0$ we have $\ymin \leq y_n$.
    Let
    \[
        \tau_0 \re{(y_1,z_1):t_1} \tau_1 \re{(y_2,z_2):t_2} \dots
    \]
    denote the corresponding run in $\T$.
    Let $i_n = \indextau{\tau_n}{\ymin}$ be the index where $\ymin-1$ appears in $\tau_n$. 
    Note that the sequence $(i_n)_{n\geq n_0}$ is decreasing. 
    Let $n_1$ be the moment where this sequence stabilises, i.e., $i_n = i_{n_1}$ for $n\geq n_1$.
    By definition of the transitions of $\T$, for $n\geq n_1$ all output priorities are $\geq 2i_{n_1}$, and priority $2i_{n_1}$ is produced every time that a letter $(\ymin, z_n)$ is read.
    We conclude that $\T$ accepts $w$.
\end{proof}

\subparagraph{$0$-freeness of automata for prefix-increasing objectives.} The fact that $W_2$ is "prefix-increasing" will be used via the following lemma. It recasts the fact that we can add $\re{\eps:1}$ transitions everywhere to automata recognising "prefix-increasing" objectives.

\begin{lemma}\label{lem:prefix-increasing}
    Let $W$ a be "prefix-increasing" objective with "memory" $\leq k$.
    There exists a deterministic "$k$-wise $\eps$-completable" automaton $\A$ recognising $W$ and an "$\eps$-completion" $\A^\eps$ of $\A$ such that $\A^\eps$ does not have any transition with priority $0$.
\end{lemma}

\begin{proof}
    First, take any automaton $\A_0$ recognising $W$.
    Then let $\A_1$ be obtained by shifting every priority from $\A_0$ by $2$.
    Clearly $\A_1$ also recognises $W$ and does not have any transition with priority $0$.
    Then apply Proposition~\ref{prop:existenceEpsComplete} to get a "$k$-blowup"  $\A$ of $\A_1$ which is "$k$-wise $\eps$-completable", and let $\A^\eps$ denote the corresponding "$\eps$-completion".
    Since $\A$ has no transition with priority $0$, the only possible such transitions in $\A^\eps$ are $\eps$-transition.
    Then remove all $\eps$-transitions with priority $0$ in $\A^\eps$, and add $\re{\eps:1}$ transitions between all pairs of states in $\A^\eps$ (in both directions).

    Clearly, the obtained automaton $\tilde \A^\eps$ is "$\eps$-complete" and has no transition with priority $0$.
    There remains to prove that it recognises $W$.
    Take an accepting run in $\tilde \A^\eps$ and observe that the priority $1$ is only seen finitely often.
    Hence from some moment on, the run coincides with a run in $\A^\eps$.
    We conclude since $W$ is "prefix-increasing".
\end{proof}

\subparagraph{Main proof: "$\eps$-completion" of the product.}
We now proceed with the proof of Theorem~\ref{thm:union}.
Using Theorem~\ref{thm:main-charac}, for $l =1,2$, we take a deterministic "$k_l$-wise $\eps$-completable" automata $\A_l$ of "index" $d_l$ recognising $W_l$, and its "$\eps$-completion" $\A_{l}^\eps$. 
For $l=2$, we assume thanks to Lemma~\ref{lem:prefix-increasing} that $\A_{2}^\eps$ does not have any transition with priority $0$.

We consider the product $\A=(\A_{1} \times \A_{2}) \casc \T$, with states $V(\A_1) \times V(\A_2) \times V(\T)$ and transitions $(q_1,q_2,\tau) \re{a:t} (q'_1,q'_2,\tau')$ if $q_1\re{a:y}q_1'$ in $\A_1$, $q_2\re{a:z}q_2'$ in $\A_2$, and $\tau \re{(y,z):t} \tau'$ in~$\T$. 
The correctness of such a construction is folklore.\footnote{$\A_1\times \A_2$ can be seen as a Muller automaton with acceptance condition the union of two parity languages. The composition with $\T$ yields a correct parity automaton, as $\T$ recognises the acceptance condition.}

\begin{claim}
    The automaton $\A$ is deterministic and recognises $W=W_{1} \cup W_{2}$.
\end{claim}

Therefore, there only remains to show the following lemma.

\begin{lemma}\label{lem:product-with-T-epsCompl}
    The automaton $\A=(\A_{1} \times \A_{2}) \casc \T$ is "$(k_1k_2)$-wise $\eps$-completable".
\end{lemma}

The "$\eps$-completion" of $\A$ will be a variant of a product of the form $\A_1^\eps \times \A_2^\eps \casc \overline \T$, where $\overline \T$ is a non-deterministic extension of $\T$ with more transitions, but which still recognises the same language.\\

\textit{The automaton $\overline \T$.}
Intuitively, we obtain $\overline \T$ by allowing to reconfigure the elements of index $>i$ in a state $\tau$ by paying an odd priority $2i-1$, as well as allowing to move elements of $\done$ to the left. We precise this idea next.

We order the states of $\T$ lexicographically, where we assume that $y < z$ for $y\in \done$ and $z\in \dtwo$.
Formally, we let $\tau < \tau'$ if for the first position $j$ where $\tau$ and $\tau'$ differ, $\tau[j]\in \done$ (and therefore necessarily $\tau'[j]\in \dtwo$).
We let $\tau[..i]$ be the prefix of $\tau$ up to (and including) $\tau[i]$.
We write $\tau <_i \tau'$ if $\tau[..i] < \tau'[..i]$. 


Let $\tau \re{(y,z):t}$ be a transition in $\T$ as above, and $i_0 = \min \{ \indtau{y}, \indtau{z}\}$ be the index determining $t$ (i.e. $t\in \{ 2i_0-1, 2i\}$). The automaton $\overline \T$ contains a transition $\tau \re{(y,z):t'} \tau'$ if:

\begin{enumerate}
    \item $t' = t$ is odd, and $\tau' \leq_{i_0-1} \tau$; or
    \item $t' = t$ is even, and $\tau' \leq_{i_0} \tau$; or
    \item $t' \in \{2i'-1, 2i'\}$ for some $i' \leq i_0$ and $\tau' <_{i'} \tau$. (Note that, if $t$ is odd, this includes all (possibly even) $t'\leq t+1$.)
\end{enumerate}

In words, we are allowed to output a small (i.e. important) priority when following a strict decrease on sufficiently small components in $\tau$. 
Note that transitions in $\T$ also belong to $\overline \T$ thanks to the rules (1) and (2).


\begin{lemma}\label{lem:overline-T}
    The automaton $\overline \T$ recognises the same language as $\T$.
\end{lemma}
\begin{proof}
    It is clear that $L(\T) \subseteq L(\overline \T)$. We show the other inclusion.

    Consider
    \[\tau_0 \re{(y_1,z_1):t_1} \tau_1 \re{(y_2,z_2):t_2} \dots , \quad \text{ an accepting run in } \overline \T.\]

    Let $\tmin=\liminf t_1t_2\dots$, which is even.
    Let $\imin = \tmin/2$ be the index responsible for producing priority $\tmin$.
    Let $n_0$ be such that $t_n \geq \tmin$ for all $n \geq n_0$.
    Observe that for all $n \geq n_0$, we have $\tau_n \geq_{(\imin-1)} \tau_{n+1}$, and therefore there is $n_1 \geq n_0$ such that the prefix $\tau_n[..\imin-1]$ is the same for all $n \geq n_1$.
    In fact, the prefix $\tau_n[..\imin]$ must be constant too, as we can only modify $\tau_n[\imin]$ using rule (1.) and that would output priority $\tmin -1$.
    Assume w.l.o.g. that $\tau_n[\imin] = y\in \doneodd$.
    Now, for each $n \geq n_1$ it must be that $y_i\geq y+1$ and it must be that $y_i = y + 1$ each time that priority $\tmin$ is produced. 
    Therefore, $\liminf y_1y_2\dots = y+1$ is even.     
\end{proof}

\textit{The "$\eps$-completion".} 
We define $\A^\eps$ as a version of the "cascade product" of $\A_{1}^\eps \times \A_{2}^\eps$ with~$\T$, in which $\eps$-transitions are also allowed to use the transitions of $\overline \T$.
We let $\cleq$ denote the preference ordering over priorities, given by
$ 1 \cleq 3 \cleq \dots\cleq d-1 \cleq d \cleq  \dots 2 \cleq 0 \text{ ($d$ even)}$.
The transitions in $\A^\eps$ are defined as follows:
\begin{itemize}
    \item For $a\in \SS$:
    $(q_1,q_2,\tau) \re{a:t} (q'_1,q'_2,\tau')$ if this transition appears in $(\A_{1} \times \A_{2})\casc \T$. 
    \item $(q_1,q_2,\tau) \re{\eps:t} (q'_1,q'_2,\tau')$ if $q_1\re{\eps:y}q_1'$ in $\A_1^\eps$, $q_2\re{\eps:z}q_2'$ in $\A_2^\eps$, and $\tau \re{(y,z):t'} \tau'$ in $\overline \T$ with $t' \cleq t$.
\end{itemize}

Note that the condition $t' \cleq t$ simply allows to output a less favorable priority, so it does not create extra accepting runs.
By definition, $\A^\eps$ has been obtained by adding $\eps$-transitions to $\A$.
It is a folklore result that composition of non-deterministic automata also preserves the language recognised, so this construction is correct.
\begin{claim}
    The automaton $\A^\eps$ recognises $W$.
\end{claim}

Therefore, we there only remains to prove the following lemma.

\begin{restatable}{lemma}{AepsCompleteUnion}\label{lem:A-eps-complete-union}
     The automaton $\A^\eps$ is $(k_{1}k_{2})$-wise $\eps$-complete. 
 \end{restatable}


First, we need a few remarks on the structure of "$\eps$-complete" automata.
We write $q \rer{\eps:x+1} q'$ to denote the conjunction of $q \re{\eps:x+1} q' $ and $q' \re{\eps:x+1} q$.
Given two states $q,q'$ in the same part of a "$k$-wise $\eps$-complete" automaton, we call breakpoint priority of $q$ and $q'$ the least even $\xbreak$ such that $q \rer{\eps:\xbreak+1} q'$ does not hold.
Note that this is a property of the unordered pair $\{q,q'\}$.
Observe also that by the definition of $\eps$-completeness, we have either $q \re{\eps:\xbreak} q'$ or $q' \re{\eps:\xbreak} q$. 
Moreover, assuming that $q \re{\eps:\xbreak} q'$, we also get that there can be no $q' \re{\eps:x} q$, for even $x$, otherwise we would accept some run labelled by $\Sigma^* \eps^\omega$; therefore for even $x \geq \xbreak$ we also have $q \re{\eps:x} q'$.
To sum up, if $\xbreak$ is the breakpoint priority of $\{q,q'\}$ and $q \re{\eps:\xbreak} q'$, then: 
\begin{itemize}
    \item $q\rer{\eps:y} q'$ for all odd $y<\xbreak$;
    \item $q\re{\eps:x} q'$ for all $x\geq \xbreak$; and
    \item there is no transition $q'\re{\eps:x} q$ for even $x$.
\end{itemize}
Finally, observe that in $\A_2^\eps$, since there are no transitions with priority $0$ (and therefore $\re{\eps:1}$ connects every ordered pair of states), breakpoint priorities are always $\geq 2$.

We are now ready to prove Lemma~\ref{lem:A-eps-complete-union}

\begin{proof}[Proof of Lemma~\ref{lem:A-eps-complete-union}]
    First observe that the $k_1$ parts of $\A_1^\eps$ and the $k_2$ parts of $\A_2^\eps$ naturally induce a partition of the states of $\A^\eps$ into $k_1k_2$ parts.
    Let $r=(q_{1},q_{2},\tau)$ and $r'=(q'_{1}, q'_{2},\tau')$ be two states in the same part of $\A^\eps$, that is, $q_{l}$ and $q'_{l}$ are in the same part in $\A_{l}^\eps$, for $l=1,2$.
    We will show that for some even output priority $\xbreak$, it holds that:
    \begin{enumerate}[1.]
        \item\label{item:small-odd} 
        $r \rer{\eps:\xbreak-1} r'$; and
        \item\label{item:large-even} 
        either $r \re{\eps:\xbreak} r'$ or $r' \re{\eps:\xbreak} r$.
    \end{enumerate}
    Note that since $r\re{\eps:t} r'$ in $\A^\eps$ implies $r \re{\eps:t'} r'$ for all $t'\cleq t$, the two points above imply that for every even $x < \xbreak$ we have $r \rer{\eps:x+1} r'$ and for every even $x \geq \xbreak$, 
    either $r \re{\eps:\xbreak} r'$ or $r' \re{\eps:\xbreak} r$.
    Therefore, this will prove that $\A^\eps$ is $(k_{1}k_{2})$-wise $\eps$-complete.

    Let $\xbreak_{1}$ and $\xbreak_{2}$ denote the breakpoint priorities of $\{q_{1},q'_{1}\}$, $\{q_{2}, q'_{2}\}$ in $\A_1^\eps$ and $\A_2^\eps$, respectively (even and $\geq 2$).
    Let $i_T$ be the largest index such that $\tau[..i_T] = \tau'[..i_T]$, with $i_T=0$ if $\tau[1]\neq \tau'[1]$.
    We distinguish two cases, depending on whether some $\xbreak_l-1$ appears in  $\tau[..i_T]$.

    \begin{description}
        \item[a) $i_T < \indtau{\xbreak_1}, \indtau{\xbreak_2}$.]
        Note that, in particular, $0<\xbreak_{1}, \xbreak_2$. 
        
        We show that in this case, we can set $\xbreak=2i_T+2$. We prove the two points above:
        \begin{itemize}
            \item[\ref{item:small-odd}.] We will find odd priorities $y_{1}$ and $y_{2}$ such that
            \[
            \underbrace{q_{1} \re{\eps:y_{1}} q'_{1}}_{\tin \A_{1}}\tand  \underbrace{q_{2} \re{\eps:y_{2}} q'_{2}}_{\tin \A_{2}} \tand \underbrace{\tau \re{(y_{1}, y_{2}):\xbreak-1} \tau'}_{\tin \overline \T},
            \]
            which gives the wanted result when applied symmetrically in the other direction.
            Consider the element $\tau[i_T+1]$ (odd), and assume w.l.o.g. that it belongs to $\done$. We let $y_1 = \tau[i_T+1]$.
            Note that $1\leq y_1< \xbreak_1$ , as $i_T+1\leq \indtau{\xbreak_1}$, and therefore we have $q_1 \rer{\eps:y_1} q'_1$. 
            We let $y_2 = \xbreak_2 - 1$;
            by definition of $\xbreak_2$ we have $q_2 \rer{\eps:y_2} q'_2$.
            As $\indextau{\tau}{y_2} = \indextau{\tau}{\xbreak_2} > \indextau{\tau}{\xbreak_1} $ we have the third wanted transition $\tau \re{(y_1,y_2):2i_T + 1} \tau'$ in $\overline \T$, as wanted.
            By flipping $\tau$ and $\tau'$ and applying the same reasoning, we get the transition $r' \re{\eps:\xbreak-1} r$, as required (note that we also have $i_T < \indextau{\tau'}{\xbreak_1}, \indextau{\tau'}{\xbreak_2}$ by definition of $i_T$).
            
            \item[\ref{item:large-even}.] We assume w.l.o.g. that $\tau' < \tau$.
            By definition of $i_T$, we have that $\tau' <_{i_T+1} \tau$. 
            Let $z = \tau[i_T+1]$ (which belong to $\dtwoodd$ if $\tau' <_{i_T+1} \tau$).
            As $z\leq \xbreak_2$, we have $q_2 \re{\eps:z} q'_2$. Also, $q_1 \re{\eps:\xbreak_1-1} q'_1$.
            Using point (3) of the definition of $\overline \T$, we have $\tau \re{(\xbreak_1-1,z):2i_T+2} \tau'$.
            We conclude that $r \re{\eps:\xbreak} r'$.
        \end{itemize}
        \item[ b) Either $\indtau{\xbreak_1}\leq i_T$ or $\indtau{\xbreak_2}\leq i_T$.]
        We assume w.l.o.g. that $\indtau{\xbreak_1} < \indtau{\xbreak_2}$.
        We let $\xbreak = 2 \indtau{\xbreak_1}$; note that $\tau \re {(\xbreak_1,\xbreak_2-1):\xbreak}$ and  $\tau \re {(\xbreak_1-1,\xbreak_2-1):\xbreak-1}$ in $\T$.
        We verify the two cases highlighted above.
        \begin{itemize}
            \item[\ref{item:small-odd}.]
            We have $q_l \rer{\eps:\xbreak_l-1} q_l$ for $l=1,2$.
            As $\tau[..i_T] = \tau'[..i_T]$, thanks to rule (1) in the definition of $\overline \T$, we have $\tau \rer{(\xbreak_1-1,\xbreak_2-1):\xbreak-1} \tau'$, and so we get $r \rer{\eps:\xbreak-1} r'$, as required.
            \item[\ref{item:large-even}.] We have one of the transitions $q_1 \re{\eps:\xbreak_1} q_1'$ or $q_1' \re{\eps:\xbreak_1} q_1$; assume we are in the first case. Since $q_2 \rer{\eps: \xbreak_2-1} q'_2$, we also have $\tau \re{(\xbreak_1,\xbreak_2-1):x} \tau'$, so we conclude that $r \re{\eps:x} r'$.\qedhere
        \end{itemize}
    \end{description}
\end{proof}

This concludes the proof of Theorem~\ref{thm:union}.

%% file: conclusion.tex
We characterised objectives in $\BCSigma$ with "memory" (or "chromatic memory") $\leq k$ as those recognised by a well-identified class of automata.
In particular, this gives the first known characterisation of ("chromatic@@mem") "memory" for "$\omega$-regular" objectives, and proves that it is decidable (in fact even in $\NP$) to compute it.
We also showed that for "$\omega$-regular objectives", "memory" (and also "chromatic memory") coincides over finite or arbitrary game graphs.
Finally, we settled (a strengthening of) Kopczyński's conjecture for $\BCSigma$ objectives.
We now discuss some directions for future work.

\subparagraph*{Exact complexity of computation of memory.} We established that computing the ("chromatic@@mem") "memory" of an "$\omega$-regular" "objective" is in $\NP$.
In fact, computing the "chromatic memory" is $\NP$-hard already for simple classes of objectives, such as Muller~\cite{Casares22Chromatic} or safety ones~\cite{BFRV23Regular}.
However, no such hardness results are known for non-chromatic memory.

\begin{question}\label{quest:ptime}
	Given a deterministic parity automaton $\A$ and a number $k$, can we decide whether the "memory" of $L(\A)$ is $\leq k$ in polynomial time?
\end{question}

This question is open already for the simpler case of regular open objectives (that is, those recognised by reachability automata).
A related question is whether one can improve on the triply exponential bound that we establish for the size of games witnessing "memory"~$>k$ (see Proposition~\ref{prop:existenceEpsComplete-finitary}).

\subparagraph{Assymetric 1-to-2-player lifts.} A celebrated result of Gimbert and Zielonka~\cite{GZ05} states that if for an "objective" $W$ both players can play optimally using positional strategies in finite games where all vertices belong to one player, then $W$ is bipositonal over finite games. This result has been extended in two orthogonal directions: to objectives where both players require finite chromatic memory~\cite{BRORV22,BRV23} (symmetric lift for memory), and to "$\omega$-regular" objectives where Eve can play positionally in $1$-player games~\cite{CO24Positional} (asymmetric lift for positionality).
In this work, we have not provided an asymmetric lift for memory, as in most cases no such result can hold.
For  $\BCSigma$ objectives, it is known to fail already for positional objectives~\cite[Section~7]{GK22Submixing}.
For non-chromatic memory, it cannot hold for "$\omega$-regular" objectives neither, because of the example described below.

\begin{restatable}{proposition}{counterexampleLift}
\label{prop:1-2-player-counterexample}
	Let $\Sigma_n = \{1,\dots, n\}$. For every $n$, the "objective" 
	\[W_n = \{w\in \SS_n^\oo \mid w \text{ contains two different letters infinitely often}\} \] 
	 has "memory" $2$ over games where Eve controls all vertices and "memory" $n$ over arbitrary~games.
\end{restatable}
\begin{proof}
	The fact that $W_n$ has "memory" $n$ follows from~\cite{DJW1997memory}.
    We show that in every game with "objective" $W_n$ and all vertices controlled by Eve, if she can win, she has a winning strategy with "memory" $2$.
    Let $G$ be such a game.
    By prefix-independence of $W_n$, we can assume that Eve wins no matter what is the initial vertex of $G$.
    For each vertex $v\in V(G)$, let $\chi_1(v)$ be the smallest element in $\SS_n$ such that there is a path starting in $v$ that produces a colour $c$, and fix one such finite path $\pi_v^1 = v \re{ uc} v'$ of minimal length.
    Let $\chi_2(v)$ be the second smallest such element (which exists, as Eve wins the game), and fix a finite path $\pi_v^2$ of minimal length producing it.
    Note that if $v\re{u} v'$ is a path in $G$, then $\chi_1(v) \leq \chi_1(v')$.

    We define a $2$-memory strategy for Eve as follows: when in a vertex $v$ and "memory state" $1$, she will take the first edge from $\pi_v^1$.
    If this edge has colour $\chi_1(v)$, we update the "memory state" to $2$, and keep it $1$ on the contrary.
    When in the "memory state" $2$, she will take the first edge from $\pi_v^1$, and update the "memory state" to $1$ if and only if this edge has colour $\chi_2(v)$.

    We show that this strategy ensures the "objective" $W$. Let $v_0\re{c_1} v_1\re{c_2}\dots$ be an infinite play consistent with this strategy.
    Let $a = \limsup \chi_1(v_i)$.
    We claim that we produce both $a$ and a colour $>a$ infinitely often.
    Let $i_0$ be large enough so that $\chi_1(v_i) > a$ for $i\geq i_0$. Note that in this case, if $v_i\re{u} v_{i'}$ is a path that does not contain colour $\chi_2(v_i)$, then  $\chi_2(v_{i'}) =  \chi_2(v_i)$.
    If in step $i$ we are in "memory state" $1$, 
    in exactly $|\pi_{v_{i}}^1|$ steps we will produce output $\chi_1(v_{i})$ and change the "memory state" to $2$.
    Likewise, by the remark above, if we are in the "memory state" $2$, in $|\pi_{v_{i}}^2|$ steps we will produce $\chi_2(v_i) > a$ and change to the "memory state" $1$.
    This concludes the proof.
\end{proof}

In his PhD thesis, Vandenhove conjectures that an asymmetric lift for "chromatic memory" holds for "$\omega$-regular" objectives~\cite[Conjecture~9.1.2]{Vandenhove23Thesis}. This question remains open.

\begin{question}
	Is there an "$\omega$-regular" "objective"  with "chromatic memory" $k$ over games where Eve controls all vertices and "chromatic memory" $k'>k$ over arbitrary~games?
\end{question}

\subparagraph*{Further decidability results for memory.} As mentioned in the introduction, many extensions of $\oo$-automata (including deterministic $\oo$-Turing machines and unambiguous $\oo$-petri nets~\cite{FSJLS22}) compute languages that are in $\BCSigma$.
We believe that our characterisation may lead to decidability results regarding the "memory" of objectives represented by these models.

\subparagraph*{Objectives in $\Delta_0^3$.} Some of the questions answered in this work in the case of $\BCSigma$ objectives are open in full generality, for instance, the generalised Kopczy\'nski's conjecture. A reasonable next step would be to consider the class $\bdelta 3 = \bsigma 3 \cap \bpi 3$.
Objectives in $\bdelta 3$ are those recognised by max-parity automata using infinitely many priorities~\cite{Skrzypczak13Colorings}.
Our methods seem appropriate to tackle this class, however, we have been unable to extend the extraction lemma (Lemma~\ref{lem:everywhere_cofinal}) used in the proof of Section~\ref{sec:existence_automata}.

%% file: appendix-structuration.tex
We give a proof of Theorem~\ref{thm:structuration}.

\structuration*

The idea is to use choice arenas, which were introduced in Ohlmann's PhD thesis~\cite[Section 3.2 in Chapter 3]{Ohlmann21PhD} for positionality, and then adapted to memory in~\cite{CO24Positional}.

\begin{proof}
    Let $H$ be the game defined as follows.
    The set of vertices is $V(G) \sqcup \powne{V(G)}$, where $\powne{V(G)}$ is the set of non-empty subsets $X$ of $V(G)$, partitioned into $V_\Adam = V(G)$ and $V_\Eve=\powne{V(G)}$.
    The initial vertex is $v_0$, the one of $G$.
    Then the edges are given by taking those of $G$, and then adding $v \rer \eps X$ whenever $v \in X$.
    The objective is $W$.

    In words, when playing in $H$, Adam follows a path of his choice in $G$, except that at any point, he may choose a set $X$ containing the current vertex $v$, and allow Eve to continue the game from any vertex of her choice in $X$.
    In some sense, Adam can hide the current vertex; this is especially true if Eve is required to play with finite memory.

    Since $G$  "satisfies" $W$, Eve wins, simply by going back to the previous vertex $v$ every time Adam picks an edge $v \re \eps X$.
    (Formally, the corresponding winning strategy has vertices $V(G) \sqcup \{(v,X) \mid v \in X\}$, projection $\pi(v)=v$ and $\pi(v,X)=X$, edges $E(G) \cup \{v \rer \eps (v,X) \mid v \in X\}$, and initial vertex $v_0$.)
    Therefore by our assumption on $W$, there is a ("chromatic@@mem") winning strategy $S$ with memory $k$, i.e. $V(S)=V(H) \times k$.
    
    Now we define $G'$ by $V(G')=V(G) \times k$, initial vertex $(v_0,m_0)$, the initial vertex of $S$, and with the edges from $E(S) \cap (V(G') \times (\Sigma \cup \{\eps\}) \times V(G'))$, together with edges $(v,m) \re \eps (v',m)$ whenever there is $X \ni v,v'$ such that $(X,m) \re \eps (v',m) \tin S$.

    We prove that $G'$  "satisfies" the conclusion of the theorem, except for well-foundedness which is dealt with below.

    \begin{description}
        \item[$G'$ is a $k$-blowup of $G$.]
        This is because $S$ is a strategy, and vertices in $V(G)$ belong to Adam in $H$.
        Therefore, for each $(v,m) \in V(G')$, and each edge $v \re c v'$ in $G$, $v \re c v'$ is also an edge in $H$, thus there is $m'$ such that $(v,m) \re c (v',m')$ is an edge in $G'$.
        \item[In the chromatic case, $G'$ is chromatic.]
        This is because $S$ is chromatic, therefore by definition of $G'$, it is chromatic with the same chromatic update function.
        \item[$G'$ is "$k$-wise $\eps$-complete".] Since it is a graph, we should prove that for each $v,v'$ and each $m$, either $(v,m) \re \eps (v',m)$ or $(v',m) \re \eps (v,m)$ in $G$.
        This follows from applying the definition of $G'$ to $X=\{v,v'\}$, since either $(X,m) \re \eps v$ or $(X,m) \re \eps v'$ is an edge in $S$ (because $S$ is without dead-ends).
        \item[$G'$ satisfies $W$.] This is because $S$ is a winning strategy, and every path in $G'$ corresponds to a path in $S$, by replacing each edge $(v,m)\re \eps(v',m)$ by $(v,m) \re \eps X \re \eps (v',m)$.
    \end{description}

    There remains a slight technical difficulty, which is that $G'$ may have $\re \eps$-cycles (in fact, it even has all $\eps$-self-loops, by applying the definition to $X$ being a singleton).
    However for each "memory state" $m$, the relation $\re \eps$ has the property that for every subset of states $X$, there is $v \in X$ (which should be seen as a minimal element) such that every $v' \in X$  "satisfies" $(v',m) \re \eps (v,m)$ in $G'$.

    Therefore for each $m$, and each $X\subseteq V(G)$ such that $X \times \{m\}$ is a strongly connected component for $\re \eps$ in $G'$, we pick an arbitrary strict well-order $\re{}$ over $X$ which extends the $\eps$-edges already present in $G$ over $X$.
    This is possible because $G$ is assumed $G$ to be "well-founded" (and it is necessary so that the obtained graph remains a $k$-blowup of $G$).
    Finally, we rewire $\eps$-edges over $X \times \{m\}$ so that they correspond to $\re{}$; it is not hard to see that the above points are not broken by this construction.
\end{proof}